\documentclass[superscriptaddress, aps,prx,twocolumn,nofootinbib]{revtex4-2}
\usepackage{amsmath}
\usepackage{amsfonts}
\usepackage{amsthm}
\usepackage{amssymb}
\usepackage{hyperref}
\usepackage{graphicx}
\usepackage{subfigure}
\usepackage{float}
\restylefloat{table}
\usepackage{siunitx}
\usepackage{bbm}
\usepackage{bm}
\usepackage{comment}
\usepackage{xcolor}
\usepackage{enumitem}
\usepackage{algorithm}
\usepackage{algorithmic}
\usepackage[capitalise]{cleveref}
\usepackage{multirow}
\usepackage{outlines}
\usepackage{mathtools}
\usepackage{tabularx}

\newcolumntype{Y}{>{\centering\arraybackslash}X}

\newtheorem{theorem}{Theorem}
\newtheorem{lemma}{Lemma}

\newtheorem{corollary}{Corollary}

\newtheorem{definition}{Definition}

\DeclareMathOperator{\Tr}{Tr}

\DeclareMathOperator{\sos}{SOS}
\DeclareMathOperator{\sa}{SA}
\DeclareMathOperator{\sossa}{SOSSA}

\DeclareMathOperator{\lcu}{LCU}
\DeclareMathOperator{\syk}{SYK}

\newcommand{\ket}[1]{| #1 \rangle} 
\newcommand{\bra}[1]{\langle #1|}
\newcommand{\ketbra}[1]{|#1 \rangle \mskip-2mu \langle #1|}
\def\Be{\textsc{{Be}}}
\def\one{{\mathchoice {\rm 1\mskip-4mu l} {\rm 1\mskip-4mu l} {\rm
1\mskip-4.5mu l} {\rm 1\mskip-5mu l}}}

\def\cO{{\cal O}}

\begin{document}
\title{Quantum simulation with sum-of-squares spectral amplification}

\author{Robbie King}
\email{robbieking@google.com}
\affiliation{Google Quantum AI, Venice, CA 90291, United States}
\affiliation{Department of Computing and Mathematical Sciences, California Institute of Technology, Pasadena, CA 91125, United States}

\author{Guang Hao Low}
\affiliation{Google Quantum AI, Venice, CA 90291, United States}

\author{Ryan Babbush}
\affiliation{Google Quantum AI, Venice, CA 90291, United States}

\author{Rolando D. Somma}
\email{rsomma@google.com}
\affiliation{Google Quantum AI, Venice, CA 90291, United States}

\author{Nicholas C. Rubin}
\email{nickrubin@google.com}
\affiliation{Google Quantum AI, Venice, CA 90291, United States}

\begin{abstract}
We present sum-of-squares spectral amplification (SOSSA), a framework for improving quantum simulation relevant to low-energy problems.
We show how SOSSA can be applied to problems like energy and phase estimation
and provide fast quantum algorithms for these problems that significantly improve over prior art. 
To illustrate the power of SOSSA in applications, we consider the Sachdev-Ye-Kitaev model, a representative strongly correlated system, and demonstrate asymptotic speedups over generic simulation methods
by a factor of the square root of the system size. Our results reinforce those observed in [G.H. Low \textit{et al.},  arXiv:2502.15882 (2025)], where SOSSA was used to achieve state-of-the-art gate costs for phase estimation of real-world quantum chemistry systems.
\end{abstract}

\maketitle

\section{Introduction}

\begin{table*}
\renewcommand{\arraystretch}{1.5}
    \begin{tabularx}{\textwidth}{cY|Y|Y|Y}
        \hline\hline
        & & \shortstack[c]{{} \\ \bf{LCU} \\ $\lambda = \lambda_{\lcu}$ \\ {}} & \shortstack[c]{{} \\ \bf{Termwise SA} \\ $\lambda = \lambda_{\lcu}$ \\ $\Delta = \Delta_{\lcu}$} & \shortstack[c]{{} \\ \bf{SOSSA} \\ $\lambda = \lambda_{\sos}$ \\ $\Delta = \Delta_{\sos}$} \\
        \hline
        (a) & {\bf Energy estimation} & $\Theta(\lambda / \epsilon)$ \cite{knill2007optimal} & \multicolumn{2}{c}{$\Theta\big(\sqrt{\Delta\lambda} / \epsilon\big)$ \Cref{lem:fast_estimation_single_term_no_prior}} \\
         & {\bf Phase estimation} & $\mathcal{O}\big(({\lambda}/{(\sqrt{p}\epsilon}))\log\frac{1}{p}\big)$ \cite{berry2024rapidinitialstatepreparation} & \multicolumn{2}{c}{$\mathcal{O}\big(\big({\sqrt{\Delta\lambda}}/{(\sqrt{p}\epsilon)}\big)\log\frac{1}{p}\big)$ \Cref{lem:fast_ground_sate_estimation_no_prior}} \\
        & {\bf Time evolution} & $\Theta(\lambda t+\log\frac{1}{\epsilon})$ \cite{low2017optimal,low2019hamiltonian} & \multicolumn{2}{c}{$\Theta\big(\sqrt{\Delta \lambda} t+\sqrt{\lambda / \Delta}\log\frac{1}{\epsilon})\big)$ \cite{zlokapa2024hamiltonian}} \\
        \hline
        (b) & {\bf SYK model} & $\lambda_{\lcu} \sim N^2$ & $\sqrt{\Delta_{\lcu} \lambda_{\lcu}} \sim N^2$ & $\sqrt{\Delta_{\sos} \lambda_{\sos}} \sim N^{\frac{3}{2}}$ \\
        \hline\hline
    \end{tabularx}
    \caption{\small \raggedright {(a)
    Simulation tasks on Hamiltonian $H$ to precision $\epsilon$ using LCU, termwise SA, and SOSSA. 
    Presented are the query complexities assuming
    access to the Hamiltonian only via the block-encodings $H / \lambda_{\lcu}$, $H_{\sa} / \sqrt{\lambda_{\lcu}}$, and $H_{\sossa} / \sqrt{\lambda_{\sos}}$, respectively. 
    The gate complexities can be determined from the gate complexities of each block-encoding, which can be different, and the additional arbitrary gates in the algorithm. 
    We show: i) estimation of the energy $E=\langle\psi|H|\psi\rangle$, ii) phase estimation of the ground state energy $E$ with initial state $\ket{\psi}$ and ground-state $\ket{\psi_0}$ satisfying $p=|\langle\psi\ket{\psi_0}|^2 >0$, and iii) time-evolution to implement $e^{-itH}$ on a state of energy at most $E$. 
    We assume $-\lambda_{\lcu} \le E \leq -\lambda_{\lcu} + \Delta_{\lcu}$ for termwise SA  and $-\beta \le E \leq -\beta + \Delta_{\sos}$ for SOSSA. The lower bound on the ground state energy  $-\beta$ is obtained in the SOS step and implies $\Delta_{\sos}\ll \Delta_{\lcu}$.
    We also provide adaptive algorithms with improved complexities that do not require knowledge of $\Delta_{\lcu}$ or $\Delta_{\sos}$. (b) Normalization factors for the SYK model demonstrate an asymptotic speedup in system size $N$ using an appropriate SOS. The gate complexities for all block-encodings is similar in this example.
    }}
    \label{tab:results}
\end{table*}

Simulating quantum many-body systems is one of the most heralded and valuable applications of quantum computing. Although efficient quantum algorithms exist for numerous quantum simulation problems \cite{kitaev1995quantum, knill2007optimal, lloyd1996universal, wiebe2010higher, berry2015simulating, low2017optimal}, further improvements are essential to fully realize robust large-scale quantum simulation when accounting for the considerable constant factor overhead associated with fault-tolerant quantum computation \cite{berry2024rapidinitialstatepreparation, rubin2023fault, babbush2021focus}. Quantum algorithm improvements are expected to arise from exploiting specific structures of the simulation instance, thus mitigating worst-case computational costs. Examples of such are exploiting symmetries in chemistry problems or locality in lattice models~\cite{loaiza2023block, Haah2019Lattice}. Here, we introduce a strategy that exploits the low-energy properties of quantum states to improve quantum simulation algorithms.  Our framework relies on two distinct algorithmic ideas: sum-of-squares (SOS) representations of Hamiltonians and spectral amplification (SA). 

The low-energy setting is one of the most promising areas in quantum simulation. This setting is relevant in the study of matter at low temperatures, including quantum phase transitions and the computation of ground-state properties. In this setting, some of the most important simulation tasks are:  i) Estimating the energy of a quantum system with respect to certain quantum state \cite{knill2007optimal}; ii) Estimating the ground state energy of a quantum system by phase estimation \cite{kitaev1995quantum}; iii) Simulating the time evolution of a quantum state under a Hamiltonian \cite{lloyd1996universal}. We will show how SOSSA significantly improves the gate complexity of generic methods for all these problems. To this end, SOSSA combines two key ideas:   first, 
it produces a suitable SOS representation of the given Hamiltonian $H$ shifted by a constant $\beta$ such that $H + \beta \one$ has a small ground-state energy, and second, it uses SA to 
essentially compute the square root of $H + \beta \one$~\cite{somma2013spectral,low2017hamiltonian,zlokapa2024hamiltonian}.  

During the SOS step, $H+\beta \one$ is processed classically and represented as a sum of positive terms. This can modify properties of $H$ that have an impact on the complexity of the simulation algorithm, such as an $\ell_1$-norm that depends on its presentation, or the ground state energy. Ideally, the SOS representation is such that $-\beta$ is a tight lower bound on the ground state energy and that the square roots of the positive terms are not too difficult to simulate. This step relates to the well-known optimization task of finding the best SOS lower bound on the ground state energy, which can be solved efficiently in classical preprocessing using semidefinite programming. During the SA step we produce a different Hamiltonian, the `square root' of $H+\beta \one$, whose eigenvalues are the square roots of those of $H + \beta \one$. 
The low-lying eigenvalues like the ground state energy are `amplified' and, for example, this can reduce the resources needed for phase estimation since the precision requirements
are now less stringent. Similarly, access to the square root of of $H+\beta \one$ allows us to
 reduce the problem of energy estimation to one of phase estimation with worse accuracy.
While the SOS step might introduce additional overheads from the complexity of the terms, the SA step has the potential to reduce it, and the method is useful if the overall combination still provides an improvement in gate complexities. Notably, this occurs in interesting systems.

In this article, we
provide a self-contained description of the SOSSA framework and show how to produce
suitable SOS representations for any (finite dimensional) Hamiltonian so that SA can be applied. This generalizes the results of Ref.~\cite{low2025fast} 
beyond quantum chemistry.
We then use this framework to construct
 improved quantum algorithms for expectation estimation
and ground-state energy (phase) estimation
in the low-energy setting. Last, we apply these quantum algorithms to the Sachdev-Ye-Kitaev (SYK) model
where we demonstrate an asymptotic advantage in terms of the system size
in energy estimation or ground-state energy estimation.
Given that the SYK model exemplifies a strongly correlated condensed matter system, these results highlight the general applicability of SOSSA to many quantum simulation problems. 
The results are summarized in~\Cref{tab:results}, where we also include prior results on time-evolution from Ref.~\cite{zlokapa2024hamiltonian} for reference.

Our SOSSA quantum algorithms improve prior art for: i) energy estimation, where in contrast with Ref.~\cite{simon2024amplifiedamplitudeestimationexploiting} we do not require an upper bound on the expectation to be estimated and the algorithm works for arbitrary (block-encoded) operators, and ii) ground-state phase estimation, where we generalize results in Ref.~\cite{low2025fast} in that we do not need an upper bound on the energy and our method works even if the overlap between the trial state and ground state is $p<1$.  Essentially, our quantum algorithms
  improve the linear  complexities of generic methods on $\lambda$ to $\sqrt{\Delta \lambda}$, where $\lambda$ is a parameter related to the norm of the Hamiltonian--the largest possible energy--and $\Delta$ is a parameter related to the low energy of the initial state. The relevant low-energy instances arise when $\Delta \ll \lambda$.

We conclude that SOSSA provides a useful framework for several quantum simulation tasks, and expect it to be applied to other systems. Furthermore, we complement these findings by providing tight lower bounds for energy and phase estimation that show our quantum algorithms are query optimal in the low-energy setting, and provide a low-depth version for expectation estimation, which scales with the standard quantum limit, that might be of independent interest for near-term applications.

Lastly, we remark that while SOSSA concerns the combination of SOS representations and SA, each has been extensively studied in prior work. SOS is used in the context of approximation algorithms \cite{goemans1995improved,pironio2010convergent}, lower bounds on ground-state energies in quantum chemistry \cite{mazziotti2006quantum, nakata2001variational}, and characterizing quantum correlations \cite{navascues2008convergent}. These methods often appear in the pseudomoment picture, which is dual to the SOS optimization that we consider here. SA was used to achieve a quadratic quantum speedup within the context of adiabatic quantum computation~\cite{somma2013spectral} and was more recently used to obtain improved quantum algorithms for simulating time evolution and phase estimation on low-energy states~\cite{low2017hamiltonian,zlokapa2024hamiltonian,low2025fast}.


\section{SOSSA}

The goal of SOSSA is to improve the gate complexity of simulation tasks in the low-energy sector. To this end, SOSSA uses SA for reducing query complexities first. Quantum signal processing and the related quantum singular value transform provide the modern machinery for quantum simulation tasks, relying on access to a block-encoding of a Hamiltonian $H$~\cite{low2019hamiltonian, gilyen2019quantum, martyn2021grand}. This is a unitary acting on an enlarged space that contains the matrix $H/\lambda$ in one of its blocks, where $\lambda \ge \|H\|$ is needed for normalization. In quantum simulation algorithms we are often interested in both, the query and gate complexities. We define the query complexity to be the number of times the block-encodings are used (including their inverses), and the gate complexity to be the total number of two-qubit gates to implement the algorithm, which includes the number of gates to implement the block-encodings in applications. 

A standard approach to construct the block-encoding uses an efficient presentation of $H$ as a linear combination of unitaries (LCU)~\cite{childs2012hamiltonian,berry2015simulating}. Suppose a Hamiltonian $H$ acts on a system of $N$ qubits and is presented as a sum of terms as
\begin{equation}
\label{eq:HLCUPaulis}
    H = \sum_{j=0}^{R-1} g_j \sigma_j \;,
\end{equation}
where $g_j \in \mathbb R$ are coefficients and $\sigma_j$ are (tensor) products of $N$ Pauli matrices (i.e.~Pauli strings). The query complexity of various quantum simulation tasks through this LCU scales linearly in the so-called $\ell_1$-norm of the Hamiltonian:
\begin{equation}
    \lambda_{\lcu} \coloneqq  \sum_{j=0}^{R-1} |g_j| \;.
\end{equation}
The LCU column of~\Cref{tab:results} presents some known results.

SA allows one to get around the linear cost in $\lambda_{\rm LCU}$  in simulation tasks involving states of low energy. SA can only be applied to Hamiltonians that are positive semidefinite. One way to achieve SA is to produce a square root of the Hamiltonian, so that the low-energy spectrum is amplified. Intuitively, this amplification arises because the square root function is steep near zero, i.e.~$\sqrt{x} \gg x$ for $x \ll 1$, where $x$ denotes a rescaled eigenvalue. See \Cref{fig:sqrt} for an illustration. For phase estimation, a precise estimate of a small eigenvalue of the original Hamiltonian can then be obtained by squaring a coarser estimate of an eigenvalue of its square root. For time evolution, the amplified low-energy spectrum allows for a more efficient approximation to the evolution operator using lower degree polynomials in quantum signal processing. In both cases, SA can lead to a lower query complexity than in the worst case. See Appendix~\ref{sec:sa} for results on SA.
\begin{figure}[htb]
\centering
\includegraphics[width=8cm]{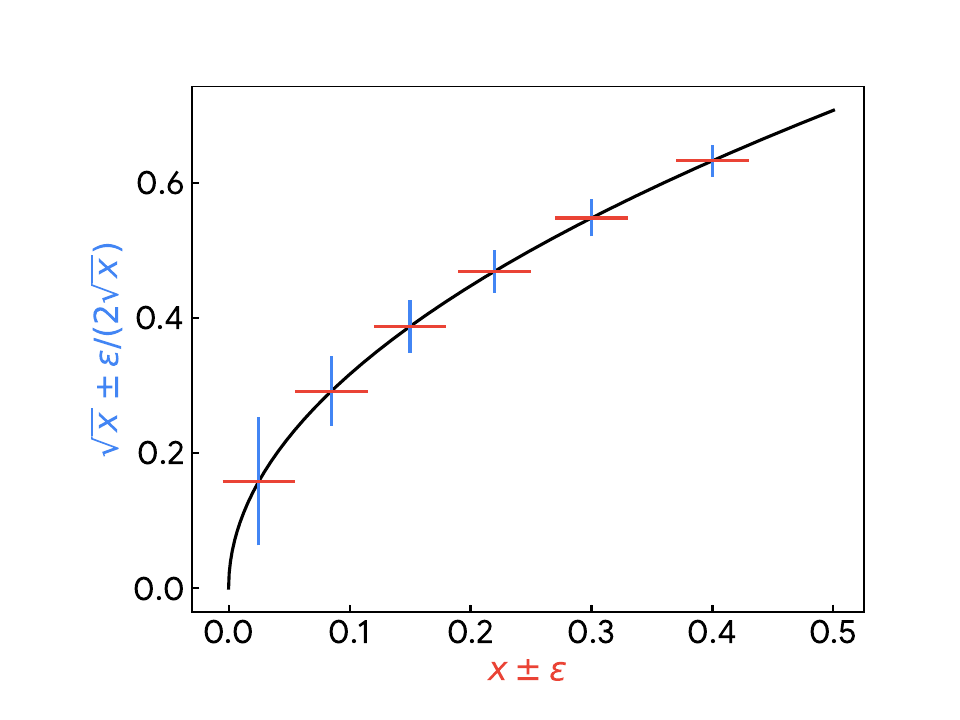}
\caption{Uncertainty propagation through the square root function. By constructing the square root of a positive semidefinite Hamiltonian, we are able to amplify the small eigenvalues due to the divergent behavior of the derivative of $\sqrt x$  near zero. Here $x \ge 0$ denotes a rescaled eigenvalue and the relevant low-energy regime occurs for $x \ll 1$.
The error bars illustrate that an estimate of an eigenvalue $E$ to precision $\epsilon$ can be obtained through an estimate of $\sqrt{E}$ to precision ${\cal O}(\epsilon/\sqrt E)$, which becomes coarser as $E$ decreases, thereby reducing the resources needed for phase estimation.
The approach can be generalized to other functions $f(x)$ whose derivatives are large or divergent as $x \rightarrow 0$ and that could arise from other related constructions, such as the quantum-walk where $f(x)=\arccos(x-1)$ results in a similar quadratic amplification~\cite{somma2013spectral,low2025fast}.} \label{fig:sqrt}
\end{figure}

The easiest way to apply SA to a Hamiltonian presented as in Eq.~\eqref{eq:HLCUPaulis} is to add a shift to each Pauli string $\sigma_j$ so that it becomes positive semidefinite; for example, $\one \pm \sigma_j  \succeq 0$, where $\one$ is the identity matrix. We call this approach `termwise SA', but also anticipate that this approach is not generally effective. During this preprocessing, the shift produces
\begin{align}\label{eq:pauli_sos_projector}
    H + \lambda_{\lcu} \one &= 2 \sum_{j=0}^{R-1} |g_j| \Pi_j \;, 
\end{align}    
where $\Pi_j :=\frac 1 2 ( \one + \text{sign}(g_j) \sigma_j)$ are simple orthogonal projectors since the Pauli strings $\sigma_j$ have eigenvalues $\pm 1$. We can define the spectral amplified operator for this specific example to be:
\begin{equation}
    H_{\sa} = \sqrt{2} \sum_{j=0}^{R-1} |g_j|^{\frac{1}{2}} \ket j \otimes \Pi_j = \sqrt{2} \begin{pmatrix}
    |g_1|^{\frac{1}{2}} \Pi_1 \\ \vdots \\ |g_R|^{\frac{1}{2}} \Pi_R
\end{pmatrix} \;.
\end{equation}
This is a rectangular matrix that acts as a square root of $H + \lambda_{\lcu} \one$ since
\begin{equation}
    H_{\sa}^\dag H_{\sa}^{\!} = H + \lambda_{\lcu} \one.
\end{equation}

Consider now a simulation problem where the corresponding  quantum states are supported in the low-energy subspace. We can introduce a parameter $\Delta_{\lcu}$ that quantifies the low-energy assumption. In particular, for energy estimation
we will assume our states of interest $\ket \psi$
to satisfy $E=\bra \psi H \ket \psi \le -\lambda_{\lcu} + \Delta_{\lcu}$.
For phase estimation
we will assume $\langle\xi|\psi\rangle = 0$ for every eigenstate $\ket \xi$ of $H$ of eigenvalue $\bra \xi H \ket \xi > -\lambda_{\lcu} + \Delta_{\lcu}$. For these problems, by using the block-encoding of $H_{\sa}/\sqrt{\lambda_{\rm LCU}}$ instead of that of $H/\lambda_{\rm LCU}$, SA allows us to achieve an improvement in the query complexity. See the termwise SA results in~\Cref{tab:results}, proven in Thms.~\ref{lem:fast_estimation_single_term_no_prior} and \ref{lem:fast_ground_sate_estimation_no_prior}. The net effect is an improved query complexity linear in $\sqrt{\Delta_{\lcu} \lambda_{\lcu}}$ rather than $\lambda_{\lcu}$. (We can think of $\sqrt{\Delta_{\lcu} \lambda_{\lcu}}$ as the new effective $\ell_1$-norm when considering the low-energy subspace.) The improvement occurs when $\Delta_{\lcu} \ll \lambda_{\lcu}$, which is possible in specific instances where the Hamiltonian $H$ is close to `frustration-free'~\cite{somma2013spectral}.

For this example, an improvement in query complexity ultimately gives an improvement in gate complexity, because the corresponding block-encodings can be implemented with similar gate costs (\Cref{thm:sa}). This would readily achieve the goal of SOSSA, however, for general and frustrated $H$, the termwise preprocessing outlined  might result in a $\Delta_{\lcu}$ that is comparable to $\lambda_{\lcu}$. This severely limits the general applicability of termwise SA, and a different kind of preprocessing is desirable to have a significant reduction in gate costs.

Our strategy is then to use a different representation of the shifted Hamiltonian as an SOS of more general operators $B_j$:
\begin{align}\label{eq:ham_sos_equality}
    H + \beta \one = \sum_{j=0}^{R-1} B_j^\dag B^{\!}_j \;.
\end{align}
Using this representation, which implies $H+\beta \succeq 0$,  gives also a path to
applying SA to general Hamiltonians. Indeed, the previous termwise SA is an example of SOS where $\beta = \lambda_{\lcu}$ and $B_j = |g_j|^{\frac{1}{2}} \Pi_j$. However, we can now consider the $B_j$'s to be linear combinations of Pauli strings $\sigma_j$ or more general operators. It is important, however, to constrain the $B_j$'s so that they can be efficiently block-encoded; see \Cref{mainsec:sos}. 

SOSSA then reduces the overall gate complexity by first reducing the query complexity via SA as much as possible. We do this by finding `good' SOS representations that have the following properties. Let 
\begin{equation}\label{maineq:sossa_sqrt}
    H_{\sossa} := \sum_{j=0}^{R-1} \ket j \otimes B_j  = \begin{pmatrix}
    B_1 \\ \vdots \\ B_R
\end{pmatrix} \;
\end{equation}
be such that $H_{\sossa}^\dag H^{\!}_{\sossa} = H + \beta \one$. Then, we wish to block-encode of $H_{\sossa} /\sqrt{\lambda_{\rm SOS}}$ efficiently, for some normalization factor $\lambda_{\rm SOS} \ge \|H+\beta \one\|$, which is often very different from $\lambda_{\rm LCU}$. Additionally, we wish for the lower bound $-\beta$ to be as close to the ground state energy of $H$ as possible.
For energy estimation, we will assume $\bra \psi H \ket \psi \le -\beta + \Delta_{\sos}$ and for
 phase estimation the low-energy state is supported on the subspace of energies in $[-\beta,-\beta+\Delta_{\sos}]$. 
These parameters give $\sqrt{\Delta_{\sos}\lambda_{\sos}}$, which determines the query complexity when using SA with the Hamiltonian $H + \beta \one$; see the SOSSA results in \Cref{tab:results}. The goal of the good SOS representation is to satisfy $\Delta_{\sos}\lambda_{\sos} \ll \Delta_{\rm LCU}\lambda_{\rm LCU}$ and hence improve upon the query complexity of termwise SA.

After obtaining improved query complexities from the SOS representation, we wish to determine the resulting gate complexities, and we need to account for the gate cost of implementing the block-encoding of $H_{\sossa} /\sqrt{\lambda_{\rm SOS}}$. Unfortunately, this gate cost might be higher than that of implementing the block-encoding of $H_{\rm SA} /\sqrt{\lambda_{\rm SA}}$, since the operators $B_j$ are more general and possibly more difficult to simulate. Nevertheless, we find that these two competing effects ---the improvement in query complexity versus the increase in the gate complexity of each query--- can still give a significant improvement overall, as demonstrated by the examples we studied.

Hence, to obtain improved gate complexities, SOSSA uses SA in combination with good SOS representations to reduce the query complexity as much as possible, even when the gate cost per query can increase.

\section{SOS optimization for SOSSA} \label{mainsec:sos}

To achieve greater SA, we would like $\Delta_{\sos}$ to be as small as possible, so we would like $-\beta$ to be a tight lower bound on the ground state energy. In this section, we show that optimizing the lower bound $-\beta$ can be achieved efficiently in classical preprocessing using semidefinite programming. Optimization of $-\beta$ alone does not directly optimize the total gate cost but can instead improve the query complexity. However, it can be a useful starting point as demonstrated in Ref.~\cite{low2025fast}. We will see in \Cref{mainsec:syk} that this optimization is sufficient to achieve asymptotic speedups in the total gate cost for the SYK model.

We can start by selecting an ansatz for the $B_j$'s as low-degree polynomials of some natural operator set $\cal B$, the `SOS algebra', of size $L$. For example, given the Hamiltonian on $N$ qubits, we can choose a constant $k \ge 1$ and let ${\cal B}$ be the set of $L=\cO(N^k)$ monomials, each being a product of $k$ Pauli operators.
Then, $B_j(\vec{b}_{j}) \in {\cal B}$ 
denotes a polynomial of Pauli strings of weight or degree $k$, where $\vec{b}_{j} \in \mathbb C^L$ are the coefficients of the polynomial.
 For a fermionic Hamiltonian, we could let ${\cal B}$ be the set of degree-$k$ products of fermionic creation and annihilation operators instead, or equivalently in the Majorana operators. Often, we may not want to include all degree-$k$ terms in ${\cal B}$, but only a subset like nearest-neighbor products when considering a system in a lattice.

Having selected the set $\cal B$, we can optimize the lower bound on the ground state energy while constraining the coefficients $\vec b_j$ so that Eq.~\eqref{eq:ham_sos_equality} is satisfied. This optimization can be cast as the following program:
\begin{align}
\min_{\vec{b}_{j} \in \mathbb{C}^{L}} \beta  \quad 
\mathrm{s.t.} \quad H + \beta \one = \sum_{j=0}^{R-1} B_j(\vec{b}_{j})^\dag B_j(\vec{b}_{j}), \nonumber
\end{align}
which can be directly translated into an SDP whose dimension is polynomial in the operator basis size $L$ and whose linear constraints are derived from the algebra of the polynomials; see \Cref{sec:sos}.
Mathematical problems of this form can be solved in polynomial time with respect to the number of variables and constraints~\cite{boyd2004convex, burer2003nonlinear,povh2006boundary}.
The resulting number of terms $R$ is related to the rank of the primal variable of the SDP and is also polynomial in $L$.

The dual problem of this SDP is the pseudomoment problem~\cite{pironio2010convergent, erdahl1978representability} where it is well understood that increasing the complexity of $B_{j}$ by, for example, allowing a larger set of monomials in $\cal B$, and thus the size of the SDP, $-\beta$ can be made arbitrarily close to the true ground-state energy, implying a smaller energy gap $\Delta_{\sos}$. 

Nevertheless, for the quantum algorithm, 
we note that the savings in complexity from a smaller $\Delta_{\sos}$
might be obscured by the higher complexity of simulating the terms $B_j$. 
Additionally, the cost of preprocessing from solving the SDP also increases with $L$, and this classical complexity 
can also be a factor of consideration in the overall optimization for a specific application.

\section{Application to SYK} \label{mainsec:syk}

Using the SOSSA framework, we observe an asymptotic speedup in phase estimation of the ground-state energy  and energy estimation over traditional approaches that use the block-encoding of the full Hamiltonian, which is constructed from an LCU presentation in terms of Pauli strings. The speedup reduces the gate complexity by a factor $\sqrt{N}$, where $N$ is the number of modes. Let $\{\gamma_1,\dots,\gamma_N\}$ be Majorana operators satisfying anticommutation relations:
\begin{equation}
    \gamma_a \gamma_b + \gamma_b \gamma_a = 2 \delta_{ab} \one \;.
\end{equation}
The SYK model is described by a fermionic Hamiltonian containing all degree-$4$ Majorana terms whose coefficients are random Gaussians; that is,
\begin{align}
\label{eq:HSYK}
H_{\syk} &= \frac{1}{\sqrt{N \choose 4}} \sum_{a,b,c,d} g_{abcd} \gamma_a \gamma_b \gamma_c \gamma_d \;, \\
g_{abcd} &\sim \mathcal{N}(0,1) \ \text{i.i.d.} \;.
\end{align}
 When applying SOSSA to the SYK model we will use degree-$2$ or quadratic Majorana operators to construct the SOS representation. 
 In \Cref{lem:syk_sos} we demonstrate that this SOS achieves a lower bound of $\beta = \cO(N)$ with high probability using random matrix theory to obtain a lower bound on the SOS dual problem.
Furthermore, the spectrum of $H_{\syk}$ is known to be contained in the interval $[-c\sqrt N, c\sqrt N]$ for some constant $c>0$, and therefore the energy of any state is such that $\Delta_{\sos} = \mathcal{O}(N)$. 

To compare the performance of SOSSA with standard approaches or termwise SA, we first note $H_{\syk}$ has number of terms scales like $\sim N^4$. Hence, the asymptotic gate complexities to construct the necessary block-encodings for the LCU, termwise SA, and SOSSA approaches is $\Theta(N^4)$ in all cases, which   are  optimal since the number of independent terms in $H_{\syk}$ is $\sim N^4$. We can thus focus on comparing the query complexities only, which are determined by $\lambda_{\lcu}$, $\sqrt{\Delta_{\lcu} \lambda_{\lcu}}$ and $\sqrt{\Delta_{\sos} \lambda_{\sos}}$. We have
\begin{align}
    \lambda_{\lcu} = \frac 1 {\sqrt{\binom{N}{4}}}\sum_{a,b,c,d} |g_{abcd}|
\end{align}
for Eq.~\eqref{eq:HSYK}, and hence
$\lambda_{\lcu} \sim N^2$ with high probability. Termwise SA does not provide any asymptotic improvement, since the scaling of the ground state energy $E$ is dominated by $\lambda_{\lcu}$ and so $\Delta_{\lcu}$ will asymptotically scale linearly with $\lambda_{\lcu}$. Nevertheless, in \Cref{sec:syk} we show that we can implement a block encoding of $H_{\sossa}$ with squared normalization factor $\lambda_{\sos} = \mathcal{O}(N^2)$. Our approach involves the double factorization technique described in \Cref{sec:double_factorization} \cite{von2021quantum}. Since the energy gap is $\Delta_{\sos}  = \mathcal{O}(N)$, we have $\sqrt{\Delta_{\sos} \lambda_{\sos}} = \cO(N^{\frac{3}{2}})$ for SOSSA, giving an asymptotic speedup by a factor of $\sim \sqrt{N}$ over LCU and termwise SA.

\subsection{Higher-degree SOS algebras}
In Ref.~\cite{hastings2022optimizing} it was also demonstrated that the degree-$2$ Majorana SOS is limited to a lower bound scaling like $\beta = \mathcal{O}(N)$ with high probability.
Reference~\cite{hastings2022optimizing}~introduces a degree-3 SOS, where the $B_j$'s in the SOS representation are cubic in the $\gamma_a$'s, and which is able to achieve a tighter lower bound to the ground state energy where $\beta = \Omega(\sqrt{N})$. Their SOS algebra $\cal B$ does not contain all degree-3 Majorana monomials, but rather uses only a particular fragment of the degree-3 terms of size $L=\mathcal{O}(N)$.

We can consider applying SOSSA with the degree-3 SOS from Ref.~\cite{hastings2022optimizing}. The energy gap will improve to $\Delta_{\sos} =\cO( \sqrt{N})$. However, the block-encoding normalization factor scales now as $\lambda_{\sos} \sim N^{\frac{7}{2}}$; see the bound on $\lambda_{\sos}$ in \Cref{sec:sossa}. We obtain $\sqrt{\Delta_{\sos} \lambda_{\sos}} = \cO(N^2)$, recovering the query complexity scaling of LCU and termwise SA. Not considering gate cost of block-encoding, which is likely higher than the termwise LCU method, we recover the termwise SA query cost suggesting higher gate complexities.

This further illustrates that optimization of $-\beta$ alone does not guarantee that we will ultimately obtain improved gate complexities.  The example analyzed here, which regarded a more involved optimization,
readily provided a factor $\sqrt{\Delta_{\sos} \lambda_{\sos}}$ that is asymptotically comparable to 
$\sqrt{\Delta_{\rm LCU} \lambda_{\rm LCU}}$.  
This is a consequence of considering more complex generators for the SOS representations.
In addition to the complexity scalings due to these factors, we often expect 
the block-encodings of the 
SOS representations to be less efficient to implement. 
These effects are important when comparing the overall gate complexity of SOSSA to that of termwise SA.

\section{Conclusion}

We described a framework for fast quantum simulation of low-energy states based on SA and SOS representations. To this end, we developed quantum algorithms for energy and phase estimation that improve over prior art and showed that SOSSA gives asymptotic improvements in gate costs with respect to traditional methods when applied to the SYK Hamiltonian. With the addition of Ref.~\cite{low2025fast}, where SOSSA already provided the state of the art for ground state energy estimation in chemical systems, we expect this framework to be generally useful and applicable to other quantum systems.

\section{Acknowledgments}
We thank Matthew Hastings for the proof of \Cref{lem:syk_sos} and Bill Huggins for helpful discussions.

\bibliographystyle{apsrev4-2}
\bibliography{refs.bib}
\onecolumngrid
\pagebreak

\appendix

\section{Quantum simulation using block-encodings and linear combination of unitaries}

In this section we discuss the LCU results in \Cref{tab:results} for energy estimation, phase estimation, and time evolution. To this end, we introduce the notion of 
a block-encoding of an operator $O \in \mathbb C^{M \times M}$. This block-encoding is a unitary $\Be[O/\lambda]  \in \mathbb C^{M' \times M'}$ acting on an enlarged space ($M' \ge M$) that has $O / \lambda$ in the first block, for some normalization constant $\lambda \ge \|O\|$,
\begin{equation}
    \Be[O/\lambda]:=  \begin{pmatrix}
        O/\lambda & \cdot \\
        \cdot & \cdot
    \end{pmatrix}\;.
\end{equation}
(We use $\|.\|$ for the spectral norm.)
When acting on systems of qubits, 
the first block is specified by the all-zero state of some ancilla register. More explicitly, 
we can write $\bra 0_{\rm a}   \Be[O/\lambda] \ket 0_{\rm a}= O/\lambda$, where `a' denotes an ancillary or `clock' register. The definition can be naturally extended to rectangular operators $O \in \mathbb C^{M \times N}$. In this case we need to invoke one projector for the $M$-dimensional space and one projector for the $N$-dimensional space. Again, associating these projectors with qubit states, we might write 
$\bra 0_{\rm a}   \Be[O/\lambda] \ket 0_{\rm a'}= O/\lambda$, where $\rm a$ and $\rm a'$ are distinct for $M \ne N$.

Block-encodings provide a natural access model for several quantum simulation algorithms. In this case we often assume access to $\Be[H/\lambda]$, the block-encoding of a Hamiltonian $H$ that models certain quantum system. The query complexity of such algorithms is determined by the number of uses to this block-encoding and, to make this complexity optimal, we would like $\lambda$ to be as small as possible (e.g., as close to $\|H\|$ as possible). This is because various quantum simulation tasks that assume access to $\Be[H/\lambda]$ have (optimal) cost depending linearly on $\lambda$, as shown by the following known results.

\begin{theorem}[Energy estimation from block-encoding \cite{knill2007optimal}]
\label{thm:EEBE}
    Let $H \in \mathbb C^{M \times M}$ be a Hermitian operator, $\lambda \ge \|H\|$ be a normalization factor, and $U$ and unitary preparing state $\ket \psi = U \ket 0$. With $\mathcal{O}(\lambda / \epsilon)$ calls to $U$, $\Be[H / \lambda]$, and their inverses, we can measure the energy  (expectation) $\bra \psi H \ket \psi$ to additive precision $\epsilon>0$.
\end{theorem}

\begin{theorem}[Phase estimation from block-encoding \cite{berry2018improved}]
\label{thm:PEBE}
    Let $H \in \mathbb C^{M \times M}$ be a Hermitian operator and $\lambda \ge \|H\|$ be a normalization factor. With $\mathcal{O}(\lambda / \epsilon)$ calls to $\Be[H / \lambda]$ and its inverse, we can perform phase (eigenvalue) estimation on $H$ to additive precision $\epsilon>0$.
\end{theorem}

\begin{theorem}[Time evolution from block-encoding \cite{low2019hamiltonian,low2017optimal}] \label{thm:time_evolution}
    Let $H \in \mathbb C^{M \times M}$ be a Hermitian operator and $\lambda \ge \|H\|$ be a normalization factor. With $\mathcal{O}(\lambda t + \log{\frac{1}{\epsilon}})$ calls to $\Be[H / \lambda]$ and its inverse, we can implement the time evolution operator $e^{-iHt}$ to additive precision $\epsilon>0$.
\end{theorem}

We briefly comment on these results. \Cref{thm:EEBE} results for performing amplitude estimation, a problem that also reduces to quantum phase estimation as shown in Ref.~\cite{knill2007optimal}. The Hamiltonian $H$ is not unitary but the amplitude estimation is done with the block-encoding, to estimate an expectation of $\Be[H/\lambda]$. For additive precision $\epsilon'$ this can be done with $\cO(1/\epsilon')$ uses of $\Be[H/\lambda]$ and the inverse. The result follows from choosing $\epsilon'=\epsilon/\lambda$.
\Cref{thm:PEBE} results from the standard use of the quantum phase estimation algorithm~\cite{kitaev1995quantum} but, instead of running phase estimation on the unitary $e^{-i H/\lambda}$, it is ran using the walk operator appearing in qubitization~\cite{low2019hamiltonian}. We describe this operator below in  ~\cref{lem:qubitization}. The benefit of doing this is that the approach does not necessitate of another routine that approximates $e^{-i H/\lambda}$ and the encoding can be done exactly and with less overhead.  \Cref{thm:time_evolution} follows from approximating the action of the evolution $e^{-iHt}$ with a finite series that uses $\Be[H/\lambda]$ and the inverse. The series can be implemented using quantum signal processing and the cost $\cO(\lambda t+ \log \frac 1 \epsilon)$ is essentially the largest degree appearing in the series.
\begin{lemma}[\label{lem:qubitization}Qubitization~\cite{low2019hamiltonian}]
Let $H \in \mathbb C^{M \times M}$ be a Hermitian operator and $\lambda \ge \|H\|$ be a normalization factor. 
Let the quantum walk operator $W:= \textsc{Ref}_{\rm a}\cdot\textsc{Be}[H/\lambda]$, where the reflection $\textsc{Ref}_{\rm a}=2\ket{0}\bra{0}_{\rm a}-\mathbbm 1$ and assume the block-encoding $\textsc{Be}[H/\lambda]$ is self-inverse.
If $\textsc{Be}[H/\lambda]$ is not self-inverse, one can always construct a self-inverse version using one query to controlled-$\textsc{Be}[H/\lambda]$ and its inverse, and two Hadamard gates.
Then for any eigenstate $\ket{\psi_j}$ of $H$ with eigenvalue $E_j$, $W$ has eigenstates $\ket{\psi_{j\pm}}$ with eigenvalues $e^{\pm i \arccos(E_j/\lambda)}$, and $\ket{\psi_j}\ket{0}_{\rm a}=\frac{1}{\sqrt{2}}(\ket{\psi_{j+}}+\ket{\psi_{j-}})$.
\end{lemma} 

We also comment on the optimality of these approaches. When given access to the block-encoding, a general method for phase estimation is known to necessitate $\Omega(\lambda/\epsilon)$ uses of the block-encoding and the inverse. Reference~\cite{mande2023tight} proves a similar lower bound. The lower bound naturally extends to \Cref{thm:EEBE} or otherwise we would be able to perform faster quantum phase estimation via energy estimation, contradicting the previous result. For time evolution, the lower bound $\tilde \Omega(t \lambda + \log \frac 1 \epsilon)$ is given in Ref.~\cite{berry2014exponential}. Since these   results are tight, we present them using $\Theta(.)$ notation in \Cref{tab:results}. We note, however, that these refer to the query complexities. For specific instances where more structure is known and can be used, the upper bounds might be improved.

\subsection{Block-encodings from linear combination of unitaries} \label{sec:LCU}

We discussed quantum simulation when having access to $\Be[H/\lambda]$ but in applications the block-encoding must be constructed from some representation of the Hamiltonian. A standard approach is based on the linear combination of unitaries (LCU) method \cite{childs2012hamiltonian,berry2015simulating}. Suppose the Hamiltonian is presented as
\begin{equation} \label{eq:H_LCUa}
    H = \sum_{j=0}^{R-1} g_j \sigma_j
\end{equation}
where the $\sigma_j$'s are unitaries, for example Pauli strings, and $g_j \in \mathbb C$ are coefficients. When applying LCU, $\lambda$ is equal to the $\ell_1$-norm of the linear combination.
\begin{definition}[$\ell_1$-norm in  LCU] \label{def:lcu_lambda}
    Define $\lambda_{\lcu}$ to be the $\ell_1$-norm
    \begin{equation}
    \label{eq:lambdalcu}
        \lambda_{\lcu} = \sum_{j=0}^{R-1} |g_j|\;.
    \end{equation}
\end{definition}

\begin{lemma}[Compilation of LCU~\cite{Childs2017Speedup}] \label{thm:lcu}
    Let $H$ be as in \Cref{eq:H_LCUa} and $\lambda_{\rm LCU} \ge \|H\|$ be as in \Cref{eq:lambdalcu}.
    Then, it is possible to construct a block-encoding of $H$, $\Be[H/\lambda_{\rm LCU}]$, using $\cO(R C_\sigma)$ quantum gates, where $C_\sigma$ is the gate complexity of the $\sigma_j$'s.  The construction generalizes to arbitrary operators $\sigma_j$, which are not necessarily unitary, as long as  $\|\sigma\| \leq 1$ and quantum circuits for the $\Be[\sigma_j]$'s are given. 
\end{lemma}
\begin{proof}
    The proof can be found in Ref.\cite{Childs2017Speedup} and here we present a version that uses the notation of this work for completeness. First, we attach a `clock register' of dimension $R$. Define the conditional unitary operator
    \begin{equation}
        \text{SELECT} = \sum_{j=0}^{R-1} \ketbra j_{\rm a} \otimes \sigma_j   \;,
    \end{equation}
    where the $\ket j_{\rm a}\in\mathbb{C}^{R}$ are basis states of $\log_2 (R)$ qubits. Also, define the state preparation unitary $\text{PREPARE}$ that performs
    \begin{align}
        \text{PREPARE} \ket 0_{\rm a} \mapsto \ket{\alpha}_{\rm a}: = \frac{1}{\sqrt{\lambda_{\lcu}}} \sum_{j=0}^{R-1} \sqrt{g_j} \ket j_{\rm a} \;.
    \end{align}
    Then, $\text{PREPARE}^{-1} \cdot \text{SELECT} \cdot \text{PREPARE}$ gives the block-encoding:
    \begin{equation}
        \bra 0 _{\rm a}\text{PREPARE}^{-1} \cdot \text{SELECT} \cdot \text{PREPARE} \ket 0_{\rm a} = \bra a _{\rm a}\text{SELECT}\ket{\alpha} _{\rm a} = \frac{1}{\lambda_{\lcu}} \sum_{j=0}^{R-1} g_j \sigma_j = \bra 0_{\rm a}\Be[H/\lambda_{\rm LCU}]\ket 0_{\rm a}.
    \end{equation}
    Note that $\ket 0_{\rm a}$ specifies the first block of the matrix. Implementing $\text{PREPARE}$ requires accessing the coefficients $g_j$ and needs $\cO(R)$ quantum gates in the worst case. Implementing $\text{SELECT}$ can be done with $\cO(R C_\sigma)$ gates, and this step dominates the cost.
\end{proof}

Often, we will be interested in instances where the $\sigma_j$ refer to Pauli strings acting on $N$ qubits involving a constant number of local Pauli operators, in which case $C_\sigma=\cO(1)$ is constant, or where the $\sigma_j$ refer to a product of a constant number of fermionic operators acting on $N$ sites, in which case, $C_\sigma=\cO(\log N)$ for an optimal fermion-to-qubit ternary-tree mapping~\cite{Jiang2020optimalfermionto}.
Alternatively, more specialized constructions~\cite{BabbushPRX18} for $\text{SELECT}$ in the Jordan-Wigner representation have gate complexity $\cO(N)$, and when $R=\Omega(N)$, the cost of block-encoding is dominated by $\text{PREPARE}$ with cost $\cO{(R\log(N))}$, i.e. $C_\sigma =\cO(\log(N))$.

\section{Quantum simulation by spectral amplification} \label{sec:sa}

In this section we introduce the basic idea of spectral amplification (SA) and provide the results that give Termwise SA and SOSSA in \Cref{tab:results}. The first version of SA was put forward in Ref.~\cite{somma2013spectral} under the name spectral gap amplification, in which the goal was to amplify the spectral gap for faster adiabatic quantum computing of frustration-free Hamiltonians rather than amplifying the whole low-energy spectrum for arbitrary Hamiltonians as we consider here. More recently,  Ref.~\cite{zlokapa2024hamiltonian} used SA to speedup time evolution for the low-energy subspace, in the block-encoding framework, which gives the third row in \Cref{tab:results}. We refer to the approach as SA because it amplifies all the small eigenvalues of eigenvectors in the low part of the spectrum of a positive semidefinite Hamiltonian. 

We give the basic results of SA using block-encodings since this is a natural framework for this approach. Later, we will discuss how to construct these block-encodings. In SA we consider first a Hamiltonian of the form
\begin{equation} \label{eq:sa_form}
    H' = \sum_{j=0}^{R-1} h_j \in \mathbb C^{M \times M} \;,
\end{equation}
where $h_j \succeq 0$ are positive semidefinite Hermitian operators satisfying $\|h_j\| \leq g_j$. We consider a factorization of the $h_j$'s so that $h_j = A_j^\dagger A_j$ for some operators $A_j$. In the following we assume $A_j \in \mathbb C^{M \times M}$ to simplify the exposition.
Next we define the corresponding `spectral amplified' operator $H_{\sa}$ by
\begin{equation} \label{eq:sga_def}
H_{\sa} := \sum_{j=0}^{R-1} \ket j_{\rm a} \otimes A_j   \in \mathbb C^{MR \times M}\;.
\end{equation}
Note that we have enlarged the space and attached an $R$-dimensional clock register `a'. Also note that $\ket j_{\rm a}$ can be replaced by any set of $R$ mutually orthogonal quantum states and does not necessarily have to be a computational basis state. As a block matrix, $H_{\sa}$ is given by
\begin{equation}
H_{\sa} = \begin{pmatrix}
    A_0 \\ \vdots \\ A_{R-1}
\end{pmatrix} \;.
\end{equation}

We will see that the main property that makes SA useful is
\begin{align}
\label{eq:samainproperty}
 H' = \sum_{j=0}^{R-1} A^\dagger_j A^{\!}_j= H_{\sa}^\dag H_{\sa}^{\!}\;.
\end{align}
That is, $H_{\sa}$ acts as a square root of $H'$ and now we can use it to produce other operators where the eigenvalues are changed. SA will then use access to $H_{\sa}$ and, to this end, we will assume access to unitaries $\Be[A_j/a_j]$, $a_j \ge \|A_j\|$, which are the block-encodings of the individual $A_j/a_j$'s.
To this end, it is useful to introduce the following block-encoding normalization factor.
\begin{definition}
    Define $\lambda$ to be
    \begin{align}
        \lambda:= \sum_{j=0}^{R-1} |a_j|^2 \;.
    \end{align}
\end{definition}

We can efficiently implement an appropriate block-encoding of $H_{\sa}$ as follows.

\begin{lemma}[Block-encoding of $H_{\rm SA}$] \label{thm:sa}
    Let $H' \succeq 0$ be a Hamiltonian of the form given in \Cref{eq:sa_form}. Assume access to the individual block-encodings $\Be[A_j/a_j]$ and their inverses. 
    Then, we can implement a block-encoding $\Be[H_{\sa} / \sqrt{\lambda}]$ with $\cO(R)$ calls to the (controlled) $\Be[A_j/a_j]$'s and additional $\cO(R)$ arbitrary two-qubit gates.
\end{lemma}
\begin{proof}
Define the unitary
\begin{equation}
\text{SELECT} := \sum_{j=0}^{R-1} \ketbra j _{\rm b} \otimes \Be\left[ \frac{A_j}{a_j}\right]  \;,
\quad 
\bra{0}_{\rm a}\Be\left[ \frac{A_j}{a_j}\right]  \ket{0}_{\rm a}= \frac{A_j}{a_j} \;,
\end{equation}
and also the unitary ${\rm PREPARE}$ that prepares the state of the clock register
\begin{equation}
{\rm PREPARE} \ket 0_{\rm b} \mapsto \ket{\alpha}_{\rm b} := \frac{1}{\sqrt{\lambda}} \sum_{j=0}^{R-1} a_j \ket j_{\rm b}\;.
\end{equation}
Then,
\begin{align}\label{eq:sa:be}\nonumber
\bra{0}_{\rm a}\text{SELECT} \cdot \text{PREPARE} \ket 0_{\rm b}\ket{0}_{\rm a} &= \bra{0}_{\rm a}\text{SELECT} \ket{\alpha}_{\rm b}  \ket{0}_{\rm a}
= \sum_{j=0}^{R-1} \frac{a_j}{\sqrt{\lambda}} \ket j_{\rm b} \otimes \bra{0}_{\rm a}\Be\left[ \frac{A_j}{a_j} \right]\ket{0}_{\rm a} 
= \sum_{j=0}^{R-1} \ket j_{\rm b} \otimes  \frac{A_j} {\sqrt{\lambda}}\\
&\equiv \bra{0}_{\rm a}\Be\left[\frac{H_{\rm SA}}{\sqrt{\lambda}} \right]\underbrace{\ket 0_{\rm b}\ket{0}_{\rm a}}_{\ket{0}_{\rm a'}}. 
\end{align}
This shows that $\text{SELECT} \cdot \text{PREPARE}$
gives the desired block encoding.
For ${\rm SELECT}$ we used the individual controlled $\Be[A_j/a_j]$ once and for ${\rm PREPARE}$
we used $\cO(R)$ arbitrary two-qubit gates.
\end{proof}
Equipped with access to $\Be[H^{}_{\rm SA}/\sqrt{\lambda}]$, we will see in \Cref{sec:sa_energy_estimation} and \Cref{sec:sa_phase_estimation} that various simulation tasks can be performed with a cost depending on $\sqrt{\Delta \lambda}$, if the state is supported on the subspace of energy at most $\Delta >0$ of $H$.
In order to construct the walk operator~\cref{lem:qubitization} underlying these tasks, it suffices to consider a block-encoding of the Hamiltonian
\begin{align}
\label{eq:HSAHermitian}
 {\bf H}_{\rm SA}:=
    \begin{pmatrix}
       {\bf 0} & H_{\sa}^{\dagger} \cr H_{\sa}^{\!} & {\bf 0}
    \end{pmatrix},
\end{align}
which has eigenvalues that coincide with the square roots of those of $H$. (This $ {\bf H}_{\rm SA}$ was used in Ref.~\cite{zlokapa2024hamiltonian} for time evolution in the low-energy subspace of $H$.) 
Since $\sqrt x \gg x$ as $x \rightarrow 0$ this readily provides the desired amplification.
However, for obtaining improved query and gate complexities, we can instead use $\Be[H^{}_{\rm SA}/\sqrt{\lambda}]$ to block-encode a shifted version of
$H'$. More precisely, in \Cref{lem:sos_block_encoding} we construct the following, self-inverse block-encoding:
\begin{align}\label{eq:t2_block_encoding}
\Be\left[\frac{H'}{\frac{1}{2}\lambda}-\one\right].
\end{align}
By construction, any low-energy state of $H' \succeq 0$ with energy $0\le E\ll\lambda$ will correspond to an eigenvalue close to $-1$ of the block-encoded operator 
$ {H'}/(\lambda/2)-\one$. Since the quantum walk operator has now eigenphases $\arccos (E_j/(\lambda/2)-1)$, similar to the square root function,
the non-linearity of $\arccos$ is what allows us to achieve SA.

\begin{lemma}[Block-encoding of $\frac{1}{2}H_{\sa}^\dagger H_{\sa}-\lambda\one$]\label{lem:sos_block_encoding}
Let $H' \succeq 0$  be a Hamiltonian and consider the factorization $H'=H_{\sa}^\dagger H^{\!}_{\sa}$, where $H_{\sa}\in\mathbb{C}^{M\times N}$.
Let  $\Be[H_{\sa}/\sqrt{\lambda}]\in\mathbb{C}^{D\times D}$ be a block-encoding of $H_{\sa}$.
Then we may block-encode either $\Be[\frac{H_{\sa}^\dagger H_{\sa}}{\frac{1}{2}\lambda}-\one]$ using one query to $\Be[H_{\sa}/\sqrt{\lambda}]$ and its inverse, $\mathcal{O}(Q\log(D/M))$ arbitrary two-qubit gates, and one ancillary qubit, or its controlled version using two ancillary qubits.
\end{lemma}
\begin{proof}
By the definition of block-encodings, given the all-zero state $\ket{0}_{\rm a'}\in\mathbb{C}^{D/N}$, $\ket{0}_{\rm a}\in\mathbb{C}^{D/M}$  and any state $\ket{\psi}_\mathrm{S}\in\mathbb{C}^{N}$,
\begin{align}
\Be\left[\frac{H_{\sa}}{\sqrt{\lambda}}\right]\ket{0}_{\rm a'}\ket{\psi}_\mathrm{S}
=
\ket{0}_{\rm a}\frac{H_{\sa}}{\sqrt{\lambda}}\ket{\psi}_\mathrm{S}+\cdots\ket{0^\perp},
\end{align}
where $|(\bra{0}_{\rm a}\otimes \one_\mathrm{S})\ket{0^\perp}|=0$.
Using quantum singular value transformations~\cite{gilyen2019quantum} with the polynomial $2x^2-1$, we can block-encode 
\begin{align}
\Be\left[\frac{H_{\sa}^\dagger H_{\sa}}{\frac{1}{2}\lambda}-\one\right]=\Be\left[\frac{H_{\sa}}{\sqrt{\lambda}}\right]^\dagger(\textsc{Ref}_{\rm a}\otimes \one)\Be\left[\frac{H_{\sa}}{\sqrt{\lambda}}\right],
\end{align}
where the reflection $\textsc{Ref}_{\rm a}:= 2\ket{0}\bra{0}_{\rm a}-\one_{\rm a}$.
This reflection can be implemented using a multi-controlled-CNOT gate which costs one ancillary qubit and $\mathcal{O}(\log(D/M))$ two-qubit gates~\cite{2017HeToffoli}.
The controlled reflection can be implemented using a multi-controlled-CNOT gate with one additional control, and so uses two additional qubits in total over that of $\Be[H_{\sa}/\sqrt{\lambda}]$.
\end{proof}

\subsection{Expectation estimation by spectral amplification} \label{sec:sa_energy_estimation}

In this section, we present new quantum algorithms summarized in~\cref{tab:expectation_estimation} that exploit SA to improve expectation estimation.
Given a block-encoding $\textsc{Be}[H/\lambda]$ of an arbitrary Hamiltonian $H$ and a unitary preparing the state $\ket{\psi}$, the expectation estimation problem is to estimate $E=\bra{\psi}H\ket{\psi}\in[-\lambda,\lambda]$ to additive error $\epsilon \ge 0$ and confidence at least $1-q\le 1$.
To simplify the presentation, we assume in the following that $E\le0$, though we emphasize that our results are symmetric about $E=0$, meaning that they depend on $-|E|$.
In contrast to the traditional approaches that give $\mathcal{O}(\lambda/\epsilon)$ query complexity, we show that SA leads to explicit scaling with the improved factor $\sqrt{\lambda(\lambda+E)}\le\lambda$, which offer greatest advantage in the limit $E/\lambda \rightarrow -1$, the `low-energy sector', and our generalized routines naturally recover previously known results without SA, which concern a different limit in which  $|E/\lambda|\ll 1$.
In the special case of amplitude estimation~\cite{brassard2002quantum} corresponding to the case where $H/\lambda=\Pi$ is a projection, prior work on `amplified amplitude estimation'~\cite{simon2024amplifiedamplitudeestimationexploiting}
already gives the improved query complexity in this setting.
Nevertheless,
our algorithms improve on this previous work  in these key areas (Row 2 of~\cref{tab:expectation_estimation}):
\begin{itemize}
\item Given an {\textit{a priori}} known upper bound $\Delta\ge \lambda+E$, the query complexity of our non-adaptive algorithm in~\cref{lem:fast_estimation_single_term} readily scales as $\mathcal{O}(\sqrt{\Delta\lambda}/\epsilon)$, improving prior work by a factor $\frac{\Delta}{\Delta-(\lambda+E)}$, which means that our results do not need the known upper bound $\Delta$ to be at least a constant multiplicative factor worse than the \textit{a priori} unknown actual value of $\lambda+E$.
\item Even without any prior knowledge of $E$, such as through the upper bound $\Delta$, the query complexity of our adaptive algorithm in~\cref{lem:fast_estimation_single_term_no_prior} scales like $\mathcal{O}(\sqrt{(\max\{\epsilon,\lambda+E\}\lambda}/\epsilon)$. The dependence on $\lambda+E$ rather than $\Delta$ is a significant improvement as $\lambda+E$ could be arbitrarily smaller than $\Delta$. Moreover, this result also naturally recovers the `super-Heisenberg' scaling of $\mathcal{O}(1/\sqrt{\epsilon})$ when $\epsilon=\Theta(\lambda+E)$ {\textit{without prior knowledge}} on $E$, which in previous work~\cite{simon2024amplifiedamplitudeestimationexploiting} required the specific condition that $\epsilon=\Theta(\lambda+E)=\Theta(\Delta)$.
\item Our results are general and apply to arbitrary $H$ that can be block-encoded, instead of only reflections or sums of reflections as in~\cite{simon2024amplifiedamplitudeestimationexploiting}. Moreover, when we instantiate with $H=2H_{\sa}^\dagger H_{\sa}-\lambda$ in~\cref{lem:sos_block_encoding}, this generalizes the case where $H$ is a reflection and $H_{\sa}$ is a projector to arbitrary block-encoded rectangular operators $H_{\sa}$. 
The low-energy sector is equivalent to assuming small
$\frac 1 \lambda \bra{\psi}H_{\sa}^\dagger H^{\!}_{\sa}\ket{\psi}$, and we give~\cref{cor:amplified_amplitude_estimation} as an example of how our general results expressed in the block-encoding framework easily recover these special cases.
\end{itemize}

\begin{table}[H]
\centering
\begin{tabular}{|c|c|c|c|c|}
\hline\hline
Year&{Reference} & Query complexity $\mathcal{O}(\cdot)$ & Needs $\Delta$ & {Comments} \\
\hline
2002&\cite{brassard2002quantum}& $\frac{\lambda}{\epsilon}\log\left(\frac{1}{q}\right)$&No& Equivalent to $H_{\sa}\propto\Pi$,
\\\cline{1-4}
2024&\cite{simon2024amplifiedamplitudeestimationexploiting}&$\frac{\sqrt{\Delta\lambda}}{\epsilon}\log\left(\frac{1}{q}\right)\left(\frac{\Delta}{\Delta-(\lambda+E)}\right)$ &Yes&  where $\Pi$ is a projector.
\\\hline
2025&\cref{lem:fast_estimation_single_term}&$\frac{\sqrt{ \Delta\lambda}}{\epsilon}\log\left(\frac{1}{q}\right)$ &Yes & Non-adaptive algorithm
\\\hline
2025&\cref{lem:fast_estimation_single_term_no_prior}&$\frac{\sqrt{\max\{\epsilon,\lambda-|E|\}\lambda}}{\epsilon}\log\left(\frac{1}{q}\right)$ &No& Adaptive algorithm
\\
\hline\hline
\end{tabular}
\caption{\label{tab:expectation_estimation}Cost of estimating an expectation $E=\bra{\psi}H\ket{\psi}\in[-\lambda,\lambda]$, to additive error $\epsilon$ and failure probability $q$, given query access to the block-encoding of $\textsc{Be}[H/\lambda]$ and a state preparation unitary for the state $\ket{\psi}$. In some cases, a known upper bound $\Delta\ge \lambda-|E|$ is needed {\em a priori}.
}
\end{table}

In the previous section we considered 
$H'=H_{\sa}^\dagger H_{\sa}$ in~\cref{eq:sa_form} and also $\|H' \|\le \lambda$. 
Using~\cref{lem:sos_block_encoding}, we can shift this to $H = 2H'-\lambda$ so that the relevant eigenvalues of $H$ are in $[-\lambda,\lambda]$.
To ease the exposition and the comparison with other methods, in the following we will assume access to $\textsc{Be}[H/\lambda]$ where $H$ is arbitrary like~\cref{eq:H_LCUa} with eigenvalues are in $[-\lambda,\lambda]$, and the low-energy sector, is that of energies at or near $-\lambda$. 
At a high level, our non-adaptive algorithm~\cref{lem:fast_estimation_single_term} works by performing quantum phase estimation on the walk operator~\cref{lem:qubitization}.
In expectation estimation, this is a walk-operator essentially on the block-encoding of a trivial $1\times 1$ operator $E=\bra{\psi}H\ket{\psi}$.
Due to the $\arccos$ non-linearity, any error $\epsilon'$ in the estimate of phase becomes a smaller square-root error $\epsilon$ in the estimate of $E$ when it is near $-\lambda$.
As the upper bound $\Delta$ is known beforehand, we know how to choose $\epsilon'$ in phase estimation to achieve the desired $\sqrt{E\Delta}$, scaling.

Achieving the optimal scaling of our~\cref{lem:fast_estimation_single_term_no_prior} without prior knowledge of $\Delta$ is significantly more challenging.
Without the upper bound $\Delta$, we cannot make a naive choice of $\epsilon'$ beforehand.
For instance, the conservative choice of $\epsilon'=\Theta({\epsilon}/\lambda)$ is guaranteed to achieve the desired accuracy, but this basically the worst-case scaling.
The next idea is to perform a binary search for $E$ using multiple $i=1,\cdots,i_\text{max}$ iterations of phase estimation to accuracy $\epsilon'_i$ that decreases geometrically and confidence $1-q_i$.
By learning an estimate $\hat{E}$ of $E$ on the fly and computing $\epsilon$ based on $\arccos(\hat{E})$, we can terminate the algorithm when the desired $\epsilon$ is achieved. 
Early termination leads to almost the desired $\tilde{\cO}(\sqrt{(E+\lambda)\lambda}/\epsilon)$ scaling, but this is still be suboptimal by a logarithmic factor as the failure probability $q_i$ accumulates over multiple steps, and without knowing $E$ beforehand, one has to make a worst-case choice of $q_i$ that assumes the number of iterations $i_\text{max}=\Theta(\log(\lambda/\epsilon))$, which leads to $\cO(\log\frac{\log(\lambda/\epsilon)}{q})$ scaling.
Our solution is to perform the binary search in two steps. First, given $\epsilon$, we perform phase estimation to error $\epsilon'=\Theta(\sqrt{\epsilon/\lambda})$. We prove that this is guaranteed to give us either an estimate $\lambda+\hat{E}$ of $\lambda+E$ to constant multiplicative error, or that $0\le \lambda+E\le\epsilon$, and so we terminate the algorithm and return $-\lambda$ as the estimate of $E$ that is correct to additive error $\epsilon$.
Second, using this estimate $\hat{E}$, we now know how many iterations of phase estimation $i_\text{max}$ are required, up to a constant additive factor.
This information is sufficient for us to make a judicious choice of $q_i$ that by a union bound, achieves the final desired confidence with the desired $\cO(\log(\frac{1}{q}))$ scaling.

To prove our first result, we require the following known result on phase estimation.
\begin{lemma}[Phase estimation with confidence intervals~\cite{knill2007optimal,berry2024rapidinitialstatepreparation,gilyen2019quantum} ]
\label{thm:block_encoding_confidence}
Let $U$ be a unitary operator with eigenstates $\ket{\psi_j}$ and corresponding eigenvalues $e^{i\theta_j}$.
Let $\ket{\psi}=\sum_{j}\sqrt{p_j}\ket{\psi_j}$ be an arbitrary superposition.
Then with $\mathcal{O}(\frac{1}{\epsilon}\log{\frac{1}{q}})$ calls to controlled-$U$, its inverse, and one copy of $\ket{\psi}$, with probability $p_j$ we estimate $\hat{\theta}_j$ such that $|(\hat{\theta}_j-\theta_j) \mod 2\pi|\le\epsilon$ with confidence $1-q$.
\end{lemma}

Next, we use the high-confidence quantum phase estimation
to obtain our first version of improved expectation or energy estimation.
\begin{theorem}[Energy estimation by spectral amplification]\label{lem:fast_estimation_single_term}
Let $H\in\mathbb{C}^{N\times N}$ be a Hamiltonian and assume access to the block-encoding $\Be[H/\lambda]\in\mathbb{C}^{D\times D}$.
Let $\ket \psi \in\mathbb{C}^{N}$ be prepared by the state preparation unitary $P$, such that $P\ket{0}_\mathrm{S}=\ket \psi$ and assume it satisfies $\bra \psi H \ket \psi \le -\lambda + \Delta$, for some known $\Delta >0$.
Then $E=\bra \psi H \ket \psi$ can be estimated to additive error $\epsilon$ and confidence $1-q$ using 
$Q=\mathcal{O}(\frac{\sqrt{\lambda\Delta}}{\epsilon}\log\frac{1}{q})$ queries to the block-encoding, state preparation unitary, and their inverses, two ancillary qubits, and $\mathcal{O}(Q\log(D/N))$ arbitrary two-qubit gates.
\end{theorem}
\begin{proof}
Observe that the following quantum circuit is a block-encoding of the $1\times 1$ matrix $(E/\lambda)$:
\begin{align}\label{eq:expectation_walk}
U:= \textsc{Be}[E/\lambda]=(\one_{\rm a}\otimes P^\dagger)\Be\left[\frac{H}{\lambda}\right](\one_{\rm a}\otimes P)\;.
\end{align}
By Qubitization in~\cref{lem:qubitization}, the quantum walk operator $W=\textsc{Ref}_{\mathrm{RS}}\cdot\Be[E/\lambda]$ has eigenvectors $\ket{\pm}:=\frac{1}{\sqrt{2}}(\ket{0}_{\mathrm{S},{\rm a}}\pm \ket{0^\perp})$ with eigenphases $\pm \arccos(E/\lambda)$.

We now apply high-confidence quantum phase estimation in~\cref{thm:block_encoding_confidence} to this quantum walk operator with input state $\ket{0}_{\mathrm{S},{\rm a}}$.
For additive error $\epsilon_{\mathrm{PEA}}$ and confidence at least $1-q$, this requires $\mathcal{O}(\frac{1}{\epsilon_{\mathrm{PEA}}}\log\frac{1}{q})$ queries to $W$.
Note that we will obtain a phase with the $+$ or $-$ sign with probability $1/2$ as the input state $\ket{0}_{\mathrm{S},{\rm a}}$ is a uniform superposition of $\ket{\pm}$.
However, this sign does not impact our estimate and can be ignored.
An error $\epsilon_\mathrm{PEA}$ in the phase translates to a smaller error $\epsilon$ in the expectation
due to the $\arccos$ function. Indeed,
using the derivative to propagate errors,
it is possible to show
\begin{align}
    \epsilon \le \sqrt{2\Delta\lambda}\; \epsilon_{\mathrm{PEA}} \;,
\end{align}
and hence choosing 
$\epsilon_\mathrm{PEA}=\epsilon/(\sqrt{2\Delta\lambda})$
suffices. This gives the desired complexity.
\end{proof}

We now show that the same scaling is achieved even without prior knowledge on the upper bound for $\bra{\psi}H|\psi\rangle$.
This requires a decision version of phase estimation called gapped phase estimation.
Given a unitary $U$ and and an eigenstate $\ket{\psi_j}$ with eigenvalue $e^{i\theta_j}$, the task is to decide with high probability whether $\theta_j$ is either in the interval $\mathcal{I}_1$ or $\mathcal{I}_2$, where these intervals are disjoint and separated by some angular gap $2\epsilon$.
Searching for $\theta_j$ like in~\cref{thm:block_encoding_confidence} can be reduced to a binary search using gapped phase estimation with optimal query complexity.
The following version of gapped phase estimation in particular uses no ancillary qubits beyond the one needed for controlled-$U$. 
\begin{lemma}[Gapped phase estimation;~\protect{Appendix D of~\cite{low2024quantumlinearalgorithmoptimal}}]
\label{thm:gapped_phase_estimation}
Let $U$ be a unitary operator and $\ket \psi$
be an eigenstate satisfying $U\ket{\psi}=e^{i\theta}\ket{\psi}$. For any eigenstate $U\ket{\psi}=e^{i\theta}\ket{\psi}$ and any $\theta_0$ and $\epsilon$ satisfying $0<\epsilon\le\theta_0\le\epsilon+\theta_0\le\frac{\pi}{2}$, ${q}>0$, there is a unitary $\textsc{Gpe}_{\epsilon,\theta_0,{q}}$ that prepares the state
\begin{align}
\textsc{Gpe}_{\epsilon,\theta_0,{q}}\ket{1}\ket{\psi}=(\alpha(\theta)\ket{0}+\beta(\theta)\ket{1})\ket{\psi},\quad|\alpha(\theta)|^2+|\beta(\theta)|^2=1,
\end{align}
where
\begin{align}
\label{eq:gpe_thm_theta_domain_1}
\forall|\theta|&\in[0,\theta_0-\epsilon],&|\alpha(\theta)|^2&\le 2{q},&|\beta(\theta)-1|&\le {q},\\
\label{eq:gpe_thm_theta_domain_2}
\forall|\theta-\pi|&\in[0,\theta_0-\epsilon],&|\alpha(\theta)|^2&\le 2{q},&|\beta(\theta)+1|&\le {q},\\
\label{eq:gpe_thm_theta_domain_2}
\forall|\theta|&\in[\theta_0+\epsilon,\pi-\theta_0-\epsilon],&|\beta(\theta)|^2&\le 2{q},&|\alpha(\theta)-1|&\le {q},
\end{align}
using $Q=\mathcal{O}(\frac{1}{\epsilon}\log\frac{1}{{q}})$ queries to controlled-$U$ and its inverse, $\mathcal{O}(Q)$ arbitrary two-qubit gates, and no ancillary qubits.
\end{lemma}
This version of gapped phase estimation produces measurement probabilities that are symmetric about $\theta=\pi/2$.
In situations where it is necessary to distinguish between the $\theta$ or $\pi-\theta$ branches, one may always perform another round of phase estimation on, say, $e^{i\pi/2}U$, using~\cref{eq:gpe_thm_theta_domain_1}.
However in the application to quantum walk operations, a subtlety is  that every eigenstate $\ket{\psi_j}$ of the block-encoded operator maps to a equal superposition of two eigenstate $\ket{\psi_{j,\pm}}$ of the walk operator with eigenvalues $e^{\pm i\arccos\theta_j}$.
Hence distinguishing between the $\pm$ branches requires the following approach of controlled gapped phase estimation.
\begin{corollary}[Controlled gapped phase estimation]\label{cor:ctrl_gapped_phase_estimation}
Let $U$ be a unitary operator. For any eigenstate $U\ket{\psi}=e^{i\theta}\ket{\psi}$ and any $\theta_0,\epsilon$ satisfying $0<\epsilon\le\theta_0\le\epsilon+\theta_0\le\frac{\pi}{2}$, ${q}>0$, there is a unitary $\textsc{CGpe}_{\epsilon,\theta_0,{q}}$ that prepares the state
\begin{align}
\textsc{CGpe}_{\epsilon,\theta_0,{q}}\ket{+}\ket{1}\ket{\psi}=(\alpha_0(\theta)\ket{01}+\alpha_1(\theta)\ket{11}+\gamma(\theta)\ket{-0})\ket{\psi},\quad|\alpha_0(\theta)|^2+|\alpha_1(\theta)|^2+|\gamma(\theta)|^2=1,
\end{align}
where
\begin{align}
\forall|\theta|&\in[0,\theta_0-\epsilon],&1-|\alpha_0(\theta)|^2&\le{q},\\
\forall|\theta-\pi|&\in[0,\theta_0-\epsilon],&1-|\alpha_1(\theta)|^2&\le{q}.
\end{align}
using $Q=\mathcal{O}(\frac{1}{\epsilon}\log\frac{1}{{q}})$ queries to controlled-$U$ and its inverse, $\mathcal{O}(Q)$ arbitrary two-qubit gates, and no ancillary qubits.
\end{corollary}
\begin{proof}
Let $\textsc{CGpe}_{\epsilon,\theta_0,{q}'}:= \ket{0}\bra{0}\otimes \one+\ket{1}\bra{1}\otimes\textsc{Gpe}_{\epsilon,\theta_0,{q}'}$ where $\textsc{Gpe}_{\epsilon,\theta_0,{q}}$ is from~\cref{thm:gapped_phase_estimation}.
Then
\begin{align}\nonumber
&\textsc{Cgpe}_{\epsilon,\theta_0,{q}'}\frac{1}{\sqrt{2}}(\ket{0}+\ket{1})\ket{1}\ket{\psi}
=\left[\frac{\ket{0}+\beta(\theta)\ket{1}}{\sqrt{2}}\right]\ket{1}\ket{\psi}+\frac{1}{\sqrt{2}}\alpha(\theta)\ket{1}\ket{0}\ket{\psi}.
\end{align}
Now apply the Hadamard gate to obtain the state
\begin{align}
\ket{\chi(\theta)}=\left[\frac{1+\beta(\theta)}{2}\ket{0}+\frac{1-\beta(\theta)}{2}\ket{1}\right]\ket{1}\ket{\psi}+\frac{1}{\sqrt{2}}\alpha(\theta)\ket{-}\ket{0}\ket{\psi}.
\end{align}
If $|\theta|\in[0,\theta_0-\epsilon]$, then from~\cref{eq:gpe_thm_theta_domain_1}, the probability that we do not measure $\ket{01}$ is at most 
\begin{align}
1-|\alpha_0(\theta)|^2
&=\frac{|1-\beta(\theta)|^2}{4}+\frac{|\alpha(\theta)|^2}{2}
\le\frac{{q}^{\prime2}}{4}+{q}'\le\frac{5}{4}{q}'\le{q},
\end{align}
where we choose ${q}'=\frac{4}{5}{q}$.
Similarly, using~\cref{eq:gpe_thm_theta_domain_2} if $|\theta-\pi|\in[0,\theta_0-\epsilon]$, the probability that we do not measure $\ket{11}$ is also at most ${q}$.
\end{proof}

We are now ready to prove the general result of this section, giving the first row of \Cref{tab:results}.

\begin{theorem}[Energy estimation by spectral amplification with no prior]\label{lem:fast_estimation_single_term_no_prior}
Let $H\in\mathbb{C}^{N\times N}$ be a Hamiltonian and assume access to the block-encoding $\Be[H/\lambda]\in\mathbb{C}^{D\times D}$.
Let $\ket \psi \in\mathbb{C}^{N}$ be prepared by the state preparation unitary $P$, such that $P\ket{0}=\ket \psi$.
Then $E=\bra{\psi}H \ket \psi$ can be estimated to any additive error $\epsilon$ and confidence $1-{q}$ using 
$Q=\mathcal{O}(\frac{\sqrt{\max\{\epsilon,\lambda-|E|\}\lambda}}{\epsilon}\log\frac{1}{{q}})$ queries to the block-encoding, state preparation unitary, and their inverses, two ancillary qubits and $\mathcal{O}(Q\log(D/N))$ arbitrary two-qubit gates.
\end{theorem}
\begin{proof}
In~\cref{lem:fast_estimation_single_term}, we constructed a unitary walk operator $W$.
Our discussion is made clearer by instead considering the walk operator $-W$, which has eigenvectors $\ket{\pm}$ and eigenphases $\theta_\pm:=\pi\mp \arccos(\alpha)$, where $\alpha:=\frac{E}{\lambda}\in[-1,1]$.
We find it convenient to define $o:=\alpha+1=E/\lambda+1\in[0,2]$. 
We can choose the principle range of $\theta_\pm$ to be symmetric about $0$.
Hence $\theta_\pm=\pm\theta\in\pm[0,\pi]$, that is, $\theta:=|\theta_+|=|\theta_-|$, and $\theta-\pi/2$ is symmetric around $\alpha=0$.
We search for $\theta$ using gapped phase estimation~\cref{thm:gapped_phase_estimation} in two steps.
In the first step, we assume that $\alpha\in[-1,0]\Rightarrow\theta\in[0,\pi/2]$ and estimate $\theta$ to additive error $\epsilon'$. 
As the bounds on the probabilities $|\alpha(\theta)|^2, |\beta(\theta)|^2$ of gapped phase estimation are also symmetric about $\pi/2$, gapped phase estimation on $-W$ cannot distinguish between cases $\theta$ or $\pi-\theta$.
However, the sign of $\beta(\theta)\approx -\beta(\pi-\theta)$ is sensitive to these cases.
In the second step, we therefore distinguish between cases using controlled-gapped phase estimation~\cref{cor:ctrl_gapped_phase_estimation}.

Let the iterations of the search be indexed by $i=0,\cdots,i_\text{max}-1$, for some $i_\text{max}>0$.
At each iteration, $\theta$ is known to be contained in an interval $\mathcal{I}_i:=[\mathcal{I}_{i,\text{l}},\mathcal{I}_{i,\text{r}}]$ with high probability $p_i$. 
Hence, $\theta\in\mathcal{I}_0=[0,\pi/2]$ with probability $p_0=1$.
At each iteration $i\ge0$, we split $\mathcal{I}_i$ into thirds
\begin{align}
\mathcal{I}_{i,\uparrow}&:=\left[\mathcal{I}_{i,\text{l}}+\frac{1}{3}|\mathcal{I}_{i}|,\mathcal{I}_{i,\text{r}}\right],
\quad
\mathcal{I}_{i,\downarrow}:=\left[\mathcal{I}_{i,\text{l}},\mathcal{I}_{i,\text{r}}-\frac{1}{3}|\mathcal{I}_{i}|\right],
\end{align}
and we will assign $\mathcal{I}_{i+1}$ to be either $\mathcal{I}_{i,\uparrow}$ or $\mathcal{I}_{i,\downarrow}$.
Note that the width $|\mathcal{I}_{i}|=\frac{\pi}{2}r^{i}$, where $r=2/3$.

We determine this assignment using gapped phase estimation $\textsc{Gpe}_{\varphi_i,\theta_i,{q}_i}$, with $\theta_i=\frac{\mathcal{I}_{i,\text{l}}+\mathcal{I}_{i,\text{r}}}{2}$ at the midpoint of $\mathcal{I}_i$ and $\varphi_i=\frac{1}{6}|\mathcal{I}_{i}|$.
From~\cref{thm:gapped_phase_estimation}, this prepares a state 
\begin{align}
\textsc{Gpe}_{\varphi_i,\theta_i,{q}_i}\ket{0}\frac{\ket{+}+\ket{-}}{2}=\frac{1}{\sqrt{2}}\sum_{x\in\{+,-\}}(\alpha(\pm\theta)\ket{0}+\beta(\pm\theta)\ket{1})\ket{x}.
\end{align}
We measure the $\{\ket{0},\ket{1}\}$ register to obtain outcome $\ket{m}$. Then we set
\begin{align}
\mathcal{I}_{i+1}=\begin{cases}
\mathcal{I}_{i,\downarrow},\quad m=0,\\
\mathcal{I}_{i,\uparrow},\quad m=1.
\end{cases}
\end{align}
The estimate of $\theta$ in $\mathcal{I}_{i}$ converts to an estimate of 
\begin{align}
\alpha&\in-\cos(\mathcal{I}_{i})=[-\cos(\mathcal{I}_{i,\text{l}}),-\cos(\mathcal{I}_{i,\text{r}})]
=
[-\cos(\theta_{i}-|\mathcal{I}_{i}|/2),-\cos(\theta_{i}+|\mathcal{I}_{i}|/2)],
\\\label{eq:fast_expectation_estimation_epsilon_1}
\epsilon_i&=|\cos(\mathcal{I}_{i})|=2\sin(\theta_i)\sin(|\mathcal{I}_i|/2),
\\
E&\in\mathcal{H}_i=[\mathcal{H}_{i,\text{l}},\mathcal{H}_{i,\text{r}}],\quad \mathcal{H}_i:= -\lambda\cos\mathcal{I}_{i},
\\
\epsilon_i'&:=|\mathcal{H}_i|=\lambda\epsilon_i,
\end{align}
where $\epsilon_i$ and $\epsilon_i'$ is the additive error to which $\alpha$ and $E$ is known respectively.

The probability that this assignment is incorrect, that is $\theta\notin\mathcal{I}_{i+1}$ is given by the maximum of the probabilities
\begin{align}
\mathrm{Pr}[m=1|\theta\in\mathcal{I}_{i,\uparrow}\backslash\mathcal{I}_{i,\downarrow}]&=\frac{|\alpha(\theta)|^2+|\alpha(-\theta)|^2}{2}\le{q}_i,
\\
\mathrm{Pr}[m=0|\theta\in\mathcal{I}_{i,\downarrow}\backslash\mathcal{I}_{i,\uparrow}]&=\frac{|\beta(\theta)|^2+|\beta(-\theta)|^2}{2}\le{q}_i.
\end{align}
We choose ${q}_{i}=\frac{6}{\pi^2}\frac{{q}}{(i_\text{max}-i+1)^{2}}$.
Then by a union bound, the failure probability of amplitude estimation after all $i_\text{max}$ steps is at most
\begin{align}
1-p_{i_\text{max}}\le\sum_{i=1}^{i_\text{max}}{q}_{i}\le\sum_{l=1}^{\infty}\frac{6}{\pi^2}\frac{{q}}{i^{2}}={q}.
\end{align}
The query complexity of all the $\textsc{GPE}$ steps is then
\begin{align}
Q&=\mathcal{O}\left(\sum_{i=0}^{i_{\text{max}}-1}\frac{1}{r^i}\log\frac{1}{{q}_i}\right)
=
\mathcal{O}\left(\sum_{i=0}^{i_{\text{max}}-1}\left(\frac{3}{2}\right)^i\left((i_\text{max}-i+1)+\log\frac{1}{{q}}\right)\right)
=
\mathcal{O}\left(\left(\frac{3}{2}\right)^{i_\text{max}}\log\frac{1}{{q}}\right).
\end{align}
When the search is complete, there are two cases of interest.
\begin{outline}[enumerate]
    \1 {$\mathcal{I}_{i_\text{max},\text{l}}=0$}: This implies that $\alpha\in[-1,-\cos{|\mathcal{I}_{i_\text{max}}|}]$, and that $o$ is small, that is $o\le\epsilon_{i_\text{max}}=2\sin^2(|\mathcal{I}_{i_\text{max}}|/2)\le\frac{\pi^2}{8}r^{2i}$. Note further that $o\le 2\sin^2(|\mathcal{I}_{i_\text{max}}|/2)$ if and only if $\mathcal{I}_{i_\text{max},\text{l}}=0$. 
    \1 {$\mathcal{I}_{i_\text{max},\text{l}}>0$}: This implies that $\alpha\ge-\cos{\mathcal{I}_{i_\text{max},\text{l}}}$ and we have an estimate of $o\ge1-\cos{\mathcal{I}_{i_\text{max},\text{l}}}$ that is bounded away from zero.
    Let $i_{\text{f}}$ be the first iteration where $\mathcal{I}_{i_\text{f},l}>0$. 
    Then $\mathcal{I}_{i_\text{f},\text{l}}=\frac{1}{2}|\mathcal{I}_{i_{\text{f}}}|$ and 
    \begin{align}
    \theta&\in\left[\frac{1}{2}|\mathcal{I}_{i_{\text{f}}}|,\frac{3}{2}|\mathcal{I}_{i_{\text{f}}}|\right]=\left[\frac{1}{2},\frac{3}{2}\right]\frac{\pi}{2}r^{i_\text{f}},\\
    o&\in\left[1-\cos\left(\frac{1}{2}|\mathcal{I}_{i_{\text{f}}}|\right),1-\cos\left(\frac{3}{2}|\mathcal{I}_{i_{\text{f}}}|\right)\right],
    \\
    \epsilon_{i_\text{f}}&=2\sin\left(|\mathcal{I}_{i_{\text{f}}}|\right)\sin\frac{|\mathcal{I}_{i_{\text{f}}}|}{2}\le4\sin^2\frac{|\mathcal{I}_{i_{\text{f}}}|}{2}=4\left(1-\cos^2\frac{|\mathcal{I}_{i_{\text{f}}}|}{2}\right)\le8\left(1-\cos\frac{|\mathcal{I}_{i_{\text{f}}}|}{2}\right)\le 8o
    \end{align}
    As $\epsilon_{i}$ decreases monotonically with $i$, this implies that we obtain a estimate of $o$ to constant multiplicative error.
    For any later iteration $i>i_\text{f}\ge0$, observe that $|\mathcal{I}_{i_\text{f}}|\le\frac{\pi}{2}\frac{2}{3}$. Hence $\sin\frac{|\mathcal{I}_{i_\text{f}}|r^{i-i_\text{f}}}{2}\le \frac{\pi}{3} r^{i-i_\text{f}}\sin\frac{|\mathcal{I}_{i_\text{f}}|}{2}$, $\sin\frac{3|\mathcal{I}_{i_\text{f}}|}{2}\le \frac{3}{2}\sin{|\mathcal{I}_{i_\text{f}}|}$, and
    \begin{align}\nonumber
    \epsilon_{i}&=2\sin\left(\theta_i\right)\sin\frac{|\mathcal{I}_{i}|}{2}
    =2\sin\left(\theta_i\right)\sin\frac{|\mathcal{I}_{i_\text{f}}|r^{i-i_\text{f}}}{2}
    \le
    2\sin\frac{3|\mathcal{I}_{i_\text{f}}|}{2}\sin\frac{|\mathcal{I}_{i_\text{f}}|r^{i-i_\text{f}}}{2}\\
    &
    \le
    \pi r^{i-i_\text{f}}\sin|\mathcal{I}_{i_\text{f}}|\sin\frac{|\mathcal{I}_{i_\text{f}}|}{2}
    =\pi r^{i-i_\text{f}}\epsilon_{i_\text{f}}\le8\pi r^{i-i_\text{f}}o.
    \end{align}
\end{outline}

We now evaluate the final error $\epsilon_{i_\text{max}}'$ of the estimate of $E$ for some choice of $i_\text{max}$.
\begin{outline}[enumerate]
\1 
Let $i_\text{max,1}:=\lceil\frac{1}{2}\log_{1/r}(\frac{\pi^2\lambda}{8\epsilon'})\rceil$. Choose $i_\text{max}=i_\text{max,1}+d$, where $d=\lceil\log_{1/r}(16\pi)\rceil$, and ${q}=\Delta_1$. The query complexity is
\begin{align}
Q_1=\mathcal{O}\left(\frac{\sqrt{\lambda}}{\sqrt{\epsilon'}}\log\frac{1}{\Delta_1}\right),
\end{align}
and we show that we either obtain an estimate of $E$ to error at most $\epsilon'$, or obtain an estimate of $(\lambda-|E|)\in[0,\lambda]$ to at most a constant multiplicative error of $\frac{1}{2}$ as follows.
\2 Case $\mathcal{I}_{i_\text{max,1},l}=0$:
\begin{align}
\epsilon'_{i_\text{max,1}}\le\lambda\frac{\pi^2}{8}r^{2i_{\text{max},1}}\le \epsilon'.
\end{align}
\2 Case $\mathcal{I}_{i_\text{max,1},l}>0$: 
\begin{align}
\epsilon'_{i_\text{max,1}+d}\le\lambda8\pi r^{d+i_{\text{max},1}-i_\text{f}}o=8\pi r^{d+i_{\text{max},1}-i_\text{f}}h\le 8\pi r^{d}(\lambda-|E|)\le \frac{1}{2}(\lambda-|E|).
\end{align}
Hence
\begin{align}
&\lambda-|E|\le \lambda-|\mathcal{H}_{i_\text{max},\text{r}}|\le \lambda-|E|+\epsilon'_{i_\text{max}}\le\frac{3}{2}(\lambda-|E|)\le\frac{3}{2}(\lambda-|\mathcal{H}_{i_\text{max},\text{r}}|).
\end{align}
\1 
Let $i_\text{max,2}:=\lceil\log_{1/r}(\frac{\sqrt{(\lambda-|\mathcal{H}_{i_\text{max},\text{r}}|)\lambda}}{c\epsilon'})\rceil$, where $c=\frac{2}{\sqrt{3}\pi}$. Choose $i_\text{max}=i_\text{max,2}$, and ${q}=\Delta_2$. 
The query complexity is
\begin{align}
Q_2=\mathcal{O}\left(\frac{\sqrt{(\lambda-|\mathcal{H}_{i_\text{max},\text{r}}|)\lambda}}{\epsilon'}\log\frac{1}{\Delta_2}\right),
\end{align}
and we obtain an estimate of $E$ to error at most $\epsilon'$ as follows.
\begin{align}\nonumber
\epsilon_{{i_\text{max}}}&=2\sin(\theta_i)\sin(|\mathcal{I}_i|/2)
\le
\sin(\theta_i)\frac{\pi}{2}r^{i_\text{max}}
\le
\sin(\arccos{(-(o+\epsilon_{{i_\text{max}}}-1))})\frac{\pi}{2}r^{i_\text{max}}
\\
&=
\sqrt{(2-o-\epsilon_{{i_\text{max}}}) (o+\epsilon_{{i_\text{max}}} ) }\frac{\pi}{2}r^{i_\text{max}}
\le
\sqrt{2\left(o+\epsilon_\text{max} \right )}
\frac{\pi}{2}r^{i_\text{max}}.
\\\nonumber
\epsilon_{{i_\text{max}}}'&=\lambda\epsilon_{{i_\text{max}}}
\le
\lambda\sqrt{2\left(o+\epsilon_{{i_\text{max}}} \right )}
\frac{\pi}{2}r^{i_\text{max}}
\le
\sqrt{2\lambda}\sqrt{(\lambda-|E|)+\epsilon_{{i_\text{max}}}'}
\frac{\pi}{2}r^{i_\text{max}}
\\
&
\le
c\frac{\pi}{\sqrt{2}}\sqrt{\frac{(\lambda-|E|)+\epsilon_{{i_\text{max}}}'}{\lambda-|\mathcal{H}_{i_\text{max},r}|}}\epsilon'
\le
c\frac{\sqrt{3}\pi}{2}\epsilon'\le \epsilon'.
\end{align}
\end{outline}

If $\mathcal{I}_{i_{{\max},2},\text{l}}\ge1-\epsilon'/2$, the estimate for $E$ in the full range $[-\lambda,\lambda]$ is already correct to error $\epsilon'$ and we terminate the algorithm.
Otherwise, we now have to determine which of $h\in[-\lambda,-\epsilon'/2)$ or $h\in(\epsilon'/2,\lambda]$ is true.
This is accomplished using controlled gapped phase estimation from~\cref{cor:ctrl_gapped_phase_estimation} using $\textsc{CGpe}_{\pi/2-\mathcal{I}_{i_\text{max},\text{r}},\pi/2,\Delta_3}$.
With failure probability at most $\Delta_3$, we measure and obtain $\ket{01}$ if $\theta\in[0,\mathcal{I}_{i_\text{max},\text{r}}]$, and obtain $\ket{11}$ if $\theta\in[\pi-\mathcal{I}_{i_\text{max},\text{r}},\pi]$.
Using the inequality
\begin{align}
\frac{\pi}{2}-\mathcal{I}_{i_\text{max},\text{r}}
&=
\frac{\pi}{2}-\arccos(|\mathcal{H}_{i_\text{max},\text{r}}|/\lambda)
=
\arcsin(|\mathcal{H}_{i_\text{max},\text{r}}|/\lambda)
\ge\frac{|\mathcal{H}_{i_\text{max},\text{r}}|}{\lambda}\ge\frac{|E|}{\lambda}\;.
\end{align}
This has query complexity
\begin{align}
Q_3=\mathcal{O}\left(\frac{\lambda}{\max(\epsilon',|E|)}\log\frac{1}{\Delta_3}\right)=
\begin{cases}
\mathcal{O}(\log\frac{1}{\Delta_3}),&|E|\ge\lambda/2,\\
\mathcal{O}(\frac{\sqrt{\lambda}\sqrt{\lambda-|E|}}{\epsilon'}\log\frac{1}{\Delta_3}),&|E|<\lambda/2.
\end{cases}
\end{align}

We then set $\Delta_1=\Delta_2=\Delta_3={q}'/3$. Then the overall query complexity to estimate $E$ to additive error $\epsilon'$ and failure probability at most ${q}'$ is
\begin{align}
Q=Q_1+Q_2+Q_3=\mathcal{O}\left(\frac{\sqrt{\lambda}}{\epsilon'}\left(\sqrt{\epsilon'}+\sqrt{\lambda-|E|}\right)\log\frac{1}{{q}'}\right).
\end{align}
The gate complexities and qubit overhead follow from the construction of the walk operator and gapped phase estimation procedure.
Finally, observe that $\sqrt{\lambda}\sqrt{\lambda-|E|}\le\sqrt{\lambda+E}\sqrt{\lambda-E}$, and also that $x+y=\mathcal{O}(\max\{x,y\})$.

\end{proof}

As a bonus, we now give a proof of amplified amplitude amplification, which is a special case of~\cref{lem:fast_estimation_single_term_no_prior} instantiated with the spectral amplified block-encoding~\cref{lem:sos_block_encoding}.
This improves on recent work, e.g.~\cite[Lemma 2]{simon2024amplifiedamplitudeestimationexploiting}, which required a known upper $\Delta\ge |\bra{\psi}\Pi|\psi\rangle|$, and had poorer query complexity $\mathcal{O}(\frac{\sqrt{\Delta}}{\epsilon}\log(\frac{1}{{q}})\frac{\Delta}{\Delta-|\bra{\psi}\Pi|\psi\rangle|})$.
\begin{corollary}[Amplified Amplitude Estimation]\label{cor:amplified_amplitude_estimation}
Let the projector $\Pi^2=\Pi$ be block-encoded by $\Be[\Pi]$.
Let the state preparation orcale $P\ket{0}=|\psi\rangle$.
Then $\bra{\psi}\Pi|\psi\rangle$ can be estimated to additive error $\epsilon$ with confidence $1-{q}$ using $\mathcal{O}\left(\frac{\sqrt{\max\{\epsilon,\bra{\psi}\Pi\ket{\psi}(1-\bra{\psi}\Pi\ket{\psi})\}}}{\epsilon}\log\frac{1}{{q}}\right)$ queries to $\Be[\Pi]$, $P$ and their inverse.
\end{corollary}
\begin{proof}
Let $H=2\Pi/\lambda-\one$, where $\lambda=1$. Using~\cref{lem:sos_block_encoding}, we block-encode $\Be[H]$ using one query to $\Be[\Pi]$ and its inverse. 
As $E=\bra{\psi}H\ket{\psi}=2\bra{\psi}\Pi\ket{\psi}-1$,~\cref{lem:fast_estimation_single_term_no_prior} states that the query complexity to estimate $\bra{\psi}H\ket{\psi}$ to additive error $\epsilon$ and confidence $1-{q}$ is 
\begin{align}
\mathcal{O}\left(\frac{\sqrt{\lambda\max\{\epsilon,\lambda-|E|\}}}{\epsilon}\log\frac{1}{{q}}\right)
=\mathcal{O}\left(\frac{(\sqrt{\max\{\epsilon,\bra{\psi}\Pi\ket{\psi}(1-\bra{\psi}\Pi\ket{\psi})\}}}{\epsilon}\log\frac{1}{{q}}\right)
\end{align}
\end{proof}
%

\subsection{Phase estimation by spectral amplification} \label{sec:sa_phase_estimation}
In this section, we present new quantum algorithms summarized in~\cref{tab:ground_state_estimation} that exploit SA to improve phase estimation of the energy of quantum ground states.
Given a block-encoding $\textsc{Be}[H/\lambda]$ and a unitary $P$ preparing the state $\ket{\psi}$, such that the overlap $p>0$ with the ground state $\ket{\psi_0}$ is known {\em a priori}, the goal is to estimate the corresponding ground state energy $E$, where $H\ket{\psi_0}=E\ket{\psi_0}$, to additive error $\epsilon$ and confidence $1-q$.
In contrast to the typical approach that scale like $\cO(\lambda/\epsilon)$ queries, we show that SA leads to explicit scaling with the improved factor $\sqrt{\lambda(\lambda+E)}\le\lambda$.
Our generalized routines naturally recover previously known results without SA, which is the $|E|\ll\lambda$ limit, in addition to achieving better qubit or query complexities.
Our results in terms of queries to $\textsc{Be}[H/\lambda]$ is most general: If there is an operator $H'=H_{\sa}^\dagger H_{\sa}$ and we are provided the block-encoding $\textsc{Be}[H_{\sa}/\sqrt{\lambda}]$, then~\cref{lem:sos_block_encoding} informs us that we may identify $H/\lambda\leftarrow (2H_{\sa}^\dagger H_{\sa}-\lambda)/\lambda$.
If the ground state energy of $H_{\sa}^\dagger H_{\sa}$ is small, we automatically realize SA as the factor $E_0=\bra{\psi_0}H_{\sa}^\dagger H_{\sa}\ket{\psi_0}-\lambda$ is close to $-\lambda$.
We note that the expectation estimation results of~\cref{sec:sa_energy_estimation} could be reduced to the specific $p=1$ case of this section.
Out algorithms improve on previous work (Row 3 of~\cref{tab:ground_state_estimation}) in a few key areas:
\begin{itemize}
    \item In the case where $\ket{\psi}$ is supported on the eigenstates with energy less than an {\textit{a priori}} known upper bound $\Delta\ge\lambda+E$, our non-adaptive algorithm~\cref{thm:sa_phase_estimation} returns an estimate of the energy of one of these eigenstates to additive error $\epsilon$ using $\mathcal{O}(\sqrt{\lambda \Delta}/\epsilon)$ queries. 
    \item In the case where we specifically want the ground state energy, for which, like previous work, we have no ${\textit{a priori}}$ known upper bound, our adaptive algorithm~\cref{lem:fast_ground_sate_estimation_no_prior} matches the scaling of previous results in the parameters $p,q$, and scale with the improved factor $\sqrt{\lambda\max\{\epsilon,\lambda+E\}}/\epsilon$.
    This factor is upper bounded by previous results, e.g. the $|E|\ll\lambda$ case. Moreover, it exhibits novel super-Heisenberg scaling like $\mathcal{O}(\sqrt{\lambda/{\epsilon}})$ when $\epsilon=\Theta(\lambda+E)$, a result which was previously unknown.
    \item Through the use of improved gapped phase estimation techniques~\cref{thm:gapped_phase_estimation}, we reduce qubit overhead from a logarithimic factor to just a constant $2$. 
    This could be relevant in practical implementations of the algorithm.
\end{itemize}

At a high-level, the proof our non-adaptive algorithm~\cref{thm:sa_phase_estimation} is very similar to that of~\cref{lem:fast_estimation_single_term} -- perform phase estimation to accuracy $\epsilon'=\lambda/\epsilon$, and propagate the $\arccos$ nonlinearity with the {\textit{a priori}} known upper bound $\Delta$.
The proof of the adaptive algorithm~\cref{lem:fast_ground_sate_estimation_no_prior} also mirrors that of~\cref{lem:fast_estimation_single_term_no_prior} in that it uses binary search by multiple iterations of gapped phase estimation, and also performs gapped phase estimation in two steps: First to a phase estimation accuracy of $\epsilon'=\sqrt{\epsilon/\lambda}$, which is guaranteed to either tell us that $E$ is $-\lambda$ to error $\epsilon$ and so we terminate the algorithm, or give us an estimate $\lambda+\hat{E}$ of $\lambda+E$ to constant multiplicative error. Second, use the estimate $\hat{E}$ to choose the number of additional phase estimation and their accuracy and confidence parameters.

\begin{table}[H]
\centering
\begin{tabular}{|c|c|c|c|c|c|}
\hline\hline
\multirow{2}{*}{Year}&\multirow{2}{*}{Reference} & \multicolumn{2}{c|}{Query complexity $\mathcal{O}(\cdot)$}  & Extra &   \\
\cline{3-4}
&&Block-encoding& State preparation&qubits& Comments\\
\hline
2017&\cite{ge2018fastergroundstatepreparation}&$\frac{\lambda^{3/2}}{\sqrt{p}\epsilon^{3/2}}\log^{\Theta(1)}\left(\frac{\lambda}{p\epsilon q}\right)$&$\frac{1}{\sqrt{p}}\sqrt{\frac{\lambda}{\epsilon}}\log^{\Theta(1)}\left(\frac{\lambda}{p\epsilon q}\right)$ &$\mathcal{O}\left(\log{\left(\frac{1}{\epsilon}\right)}\right)$&
\\\hline
2020&\cite{Lin2020nearoptimalground}&$\frac{\lambda}{\sqrt{p}\epsilon}\log\left(\frac{\lambda}{\epsilon}\right)\log\left(\frac{1}{p}\right)\log\left(\frac{\log(\lambda/\epsilon)}{q}\right)$&$\frac{1}{\sqrt{p}}\log\left(\frac{\lambda}{\epsilon}\right)\log\left(\frac{\log(\lambda/\epsilon)}{q}\right)$ & $\mathcal{O}\left(\log\left({\frac{1}{p}}\right)\right)$&
\\\hline
2024&\cite{berry2024rapidinitialstatepreparation}&$\frac{\lambda}{\sqrt{p}\epsilon}\log\left(\frac{1}{p}\right)\log\left(\frac{1}{q}\right)$&$\frac{1}
{\sqrt{p}}\log\left(\frac{\lambda}{\epsilon}\right)\log\left(\frac{\log(\lambda/\epsilon)}{q}\right)$& $\mathcal{O}\left(\log{\left(\frac{\lambda\log(1/q)}{\epsilon p}\right)}\right)$&
\\\hline
2025&\cite{low2025fast},\cref{thm:sa_phase_estimation} &$\frac{\sqrt{\Delta\lambda}}{\epsilon}\log\left(\frac{1}{q}\right)$&$1$ & $\mathcal{O}\left(\log(1/\epsilon)\right)$& $((\lambda+E)\;\text{of}\;\ket{\psi})\le\Delta$, $p=1$.
\\\hline
2025&\cref{lem:fast_ground_sate_estimation_no_prior} &$\frac{\sqrt{\max\{\epsilon,\lambda+E\}\lambda}}{\sqrt{p}\epsilon}\log\left(\frac{1}{p}\right)\log\left(\frac{1}{q}\right)$&$\frac{1}
{\sqrt{p}}\log\left(\frac{\lambda}{\epsilon}\right)\log\left(\frac{\log(\lambda/\epsilon)}{q}\right)$ & 2&
\\
\hline\hline
\end{tabular}
\caption{\label{tab:ground_state_estimation}Cost of estimating an {\textit a priori} unknown ground state energy $E$ of $H$ to additive error $\epsilon$ and failure probability $q$, given query access to the block-encoding of $\textsc{Be}[H/\lambda]$ and a state preparation unitary for a trial state with overlap $\sqrt{p}$ with the ground state.}
\end{table}

First we consider phase estimation, the result in the first line of \Cref{tab:results}. This is the key result which enabled the improved quantum chemistry compilations in Ref.~\cite{low2025fast}.
The following proof is based on ${\bf H}_{\sa}$ for simplicity; a similar proof follows using $H_{\sa}$ combined with~\cref{lem:sos_block_encoding}.

\begin{theorem}[Phase estimation with SA] \label{thm:sa_phase_estimation}
    Let $H \succeq 0$ be a Hamiltonian of the form given in \Cref{eq:sa_form}. Let $\ket \psi$ be a low-energy eigenstate supported on the subspace of energy of $H$ at most $\Delta > 0$. We can perform phase estimation of $H$ on the state $\ket \psi$ to additive precision $\epsilon$ with $\mathcal{O}(\sqrt{\Delta \lambda} / \epsilon)$ calls to $\Be[H_{\sa} / \sqrt{\lambda}]$, the block-encoding of $H_{\sa}$, and its inverse.
\end{theorem}
\begin{proof}
    Consider the related problem of performing quantum phase estimation on the Hamiltonian ${\bf H}_{\rm SA}$ of Eq.~\eqref{eq:HSAHermitian} within additive precision $\epsilon >0$. This can be done by using the block-encoding $\Be[{\bf H}_{\rm SA}/\sqrt \lambda]$, and this block encoding can be easily implemented with one call to (controlled) $\Be[ {H}_{\rm SA}/\sqrt \lambda]$ and one call to its inverse. (An explicit construction of this block-encoding is in Ref.~\cite{zlokapa2024hamiltonian}.) Quantum phase estimation necessitates $\cO(\sqrt \lambda/\epsilon)$ uses of these block-encodings for this precision, which is known as the `Heisenberg limit'~\cite{berry2018improved,knill2007optimal}. Next we note that 
    \begin{align}
     ({\bf H}_{\rm SA})^2 = \begin{pmatrix}
         H & {\bf 0} \cr {\bf 0} & H_{\rm SA}^{\!}
         H_{\rm SA}^{\dagger}
     \end{pmatrix}    \;,
    \end{align}
    implying that, if $E \ge 0$ is the desired eigenvalue of $H$ (i.e., $H \ket \psi = E \ket \psi)$, $\pm \sqrt E$ are also eigenvalues of ${\bf H}_{\rm SA}$. Indeed, the (at most) two dimensional subspace spanned by $\{\ket \psi \ket 0 , {\bf H}_{\rm SA} \ket \psi \ket 0\}$ is invariant under the action of ${\bf H}_{\rm SA}$ and is the one that gives rise to the $\pm \sqrt {E}$. Hence, phase estimation using ${\bf H}_{\rm SA}$ with initial state $\ket \psi \ket 0$ will produce an estimate of $\sqrt E$ or $-\sqrt E$.
    
    Let $\zeta$ be such an estimate; for example consider $\sqrt E$ in which case we have $|\zeta - \sqrt E| \le \epsilon$. (The analysis can be applied to the case of $-\sqrt E$.) If we take the square we obtain
    \begin{align}
    E - 3 \sqrt E \epsilon  \le  (\sqrt E - \epsilon)^2 \le  \zeta ^2 \le (\sqrt E + \epsilon)^2 \le E + 3 \sqrt E \epsilon \;,
    \end{align}
    where we assumed the precision to satisfy $\epsilon \le \sqrt E$. Hence, if we set $\epsilon \le \epsilon/(3 \sqrt E)$ we can obtain $E$ within additive precision $\epsilon$, which is the desired goal.
    
    To this end, we run quantum phase estimation with ${\bf H}_{\rm SA}$, initial state $\ket \psi \ket 0$, and additive precision $\epsilon \le \epsilon/(3 \sqrt E)$. The number of calls to (controlled) $\Be[ {H}_{\rm SA}/\sqrt \lambda]$ and its inverse will be $\cO(\sqrt{E \lambda}/\epsilon )$. The result follows from the assumption $E \le \Delta$.
\end{proof}

Previous work~\cite{berry2024rapidinitialstatepreparation} summarized in~\cref{tab:ground_state_estimation} solves this problem with a query complexity of 
\begin{align}
\label{eq:rapid_state_prep_BE_queries}
\text{Queries to}\;\Be\left[\frac{H}{\lambda}\right]&=\mathcal{O}\left(\frac{\lambda}{\sqrt{p}\epsilon}\log\left(\frac{1}{p}\right)\log\left(\frac{1}{q}\right)\right),
\\
\text{Queries to}\;P&=\mathcal{O}\left(\frac{1}
{\sqrt{p}}\ln\left(\frac{\lambda}{\epsilon}\right)\ln\left(\frac{\ln\left({\lambda}/{\epsilon}\right)}{\sqrt{q}}\right)\right).
\end{align}
The key idea is reducing the estimation of $E$ to a fuzzy binary search problem.
Roughly speaking, at each iteration $i=0,\cdots,i_{\text{max}}-1$, an interval $\mathcal{I}_i$ containing $E$ with high confidence $1-q_i$ is found with query complexity  $\mathcal{O}(\frac{1}{|\mathcal{I}_i|}\log(\frac{1}{q_i}))$.
As $i$ increases, the width $|\mathcal{I}_i|$ is chosen to decrease geometrically.
Then by a union bound over, the final iteration estimates $E_{\sos}$ with additive error $|\mathcal{I}_i|/2$, failure probability $q=\sum_{i=0}^{i_{\text{max}}-1}q_i$.
Now suppose that $E$ is close to $-\lambda$, such as by a appropriate $H_{\sa}$ construction.
Then a modification of exploiting the $\arccos$ walk nonlinearity in the original proof of~\cref{eq:rapid_state_prep_BE_queries} leads to an improved spectral amplified query complexity.
\begin{theorem}[Fast ground state energy estimation by spectral amplification with no prior]\label{lem:fast_ground_sate_estimation_no_prior}
Let any Hermitian $H\in\mathbb{C}^{N\times N}$ have ground state $\ket{\psi_0}$ with unknown energy $E$.
Let $|\psi\rangle$ be prepared by the state preparation unitary $P\ket{0}_\mathrm{S}=|\psi\rangle$ such that the overlap $|\langle \psi|\psi_0\rangle|\ge \sqrt{p}$.
Let  $\Be[H/\lambda]\in\mathbb{C}^{D\times D}$ be a block-encoding of $H_{\sa}$.
Then if $E\in[-\lambda,0]$, it can be estimated to any additive error $\epsilon$ and confidence $1-q$ using 
\begin{align}
\label{eq:rapid_state_prep_BE_SA_queries}
\text{Queries to}\;\Be\left[\frac{H}{\lambda}\right]&=\mathcal{O}\left(\frac{\sqrt{\max\{\epsilon,\lambda-|E|\}\lambda}}{\sqrt{p}\epsilon}\log\left(\frac{1}{p}\right)\log\left(\frac{1}{q}\right)\right),
\\
\text{Queries to}\;P&=\mathcal{O}\left(\frac{1}
{\sqrt{p}}\log\left(\frac{\lambda}{\epsilon}\right)\log\left(\frac{\log(\lambda/\epsilon)}{q}\right)\right),
\end{align}
and their inverses, and two ancillary qubits and $\mathcal{O}(Q\log(D/N))$ arbitrary two-qubit gates.
\end{theorem}

\begin{proof}[Proof of~\cref{eq:rapid_state_prep_BE_SA_queries}.]
We find it convenient to work in units of phase.
Then using qubitization~\cref{lem:qubitization} on $\textsc{Be}[H/\lambda]$, we form a quantum walk $W$.
For each eigenstate $\ket{\psi_j}$ of $H\ket{\psi_j}=E_{j}\ket{\psi_j}$, let $\ket{\psi_j}=\frac{1}{\sqrt{2}}(\ket{\psi_j^+}+\ket{\psi_j^-})$, where $\ket{\psi_j^\pm}$ are all mutually orthogonal.
Then the eigenphases of the quantum walk are 
\begin{align}
-W\ket{\psi_j^{\pm}}=\exp\left(i\Theta_{\pm}\right),
\quad
\Theta_{j,\pm}=\pi\mp i \arccos\left(\alpha_j\right),
\quad\alpha_j=o_j-1,
\quad
o_j=\frac{E_j+\lambda}{\lambda},
\end{align}
where the minimum energy is $E:= E_{0}$.
Let $\alpha:= \alpha_0$, $o:= o_0\in[0,2]$, and $\theta:=|\Theta_{0,\pm}|\in[0,\pi]$, where we choose the principal range of $\Theta_{j,\pm}\in[-\pi/2,\pi/2]$.
Let us further assume that $o\in[0,1]$, hence $\theta\in \mathcal{I}_0:=[0,\pi/2]$.

Let us now perform a search for $\theta$.
Let the iterations of the search be indexed by $i=0,\cdots,i_\text{max}-1$, for some $i_\text{max}>0$
At each iteration, $\theta$ is known to be contained in an interval $\mathcal{I}_i:=[\mathcal{I}_{i,\text{l}},\mathcal{I}_{i,\text{r}}]$ with high probability $p_i$. 
At each iteration $i\ge0$, we split $\mathcal{I}_i$ into thirds
\begin{align}
\mathcal{I}_{i,\uparrow}&:=\left[\mathcal{I}_{i,\text{l}}+\frac{1}{3}|\mathcal{I}_{i}|,\mathcal{I}_{i,\text{r}}\right],
\quad
\mathcal{I}_{i,\downarrow}:=\left[\mathcal{I}_{i,\text{l}},\mathcal{I}_{i,\text{r}}-\frac{1}{3}|\mathcal{I}_{i}|\right],
\end{align}
and we will assign $\mathcal{I}_{i+1}$ to be either $\mathcal{I}_{i,\uparrow}$ or $\mathcal{I}_{i,\downarrow}$.
Note that the width $|\mathcal{I}_{i}|=\frac{\pi}{2}r^{i}$, where $r=2/3$.

We determine this assignment using gapped phase estimation $\textsc{Gpe}_{\varphi_i,\theta_i,\delta_i}$~\cref{thm:gapped_phase_estimation}, with $\theta_i=\frac{\mathcal{I}_{i,\text{l}}+\mathcal{I}_{i,\text{r}}}{2}$ at the midpoint of $\mathcal{I}_i$ and $\varphi_i=\frac{1}{6}|\mathcal{I}_{i}|$.
This prepares a state 
\begin{align}\nonumber
&\textsc{Gpe}_{\varphi_i,\theta_i,\delta_i}\ket{0}\ket{0}\ket{\psi_t}
\\\nonumber
&=\textsc{Gpe}_{\varphi_i,\theta_i,\delta_i}\ket{0}\left(\sqrt{\frac{p}{2}}(\ket{\psi_{0,+}}+\ket{\psi_{0,-}})+\sum_{j>0}\cdots\sum_{\pm}\ket{\psi_{j,\pm}}\right)
\\
&=\sqrt{\frac{p}{2}}\sum_{\pm}(\alpha(\pm\theta)\ket{0}+\beta(\pm\theta)\ket{1})\ket{\psi_{0,\pm}})+\sum_{j>0}\cdots\sum_{\pm}(\alpha(\Theta_{j,\pm})\ket{0}+\beta(\Theta_{j,\pm})\ket{1})\ket{\psi_{j,\pm}}.
\end{align}
Let us measure the $\{\ket{0},\ket{1}\}$ register to obtain outcome $\ket{m}$. 
Let $\kappa=|(\ket{0}\bra{0}\otimes I)\textsc{Gpe}_{\varphi_i,\theta_i,\delta_i}\ket{0}\ket{0}\ket{\psi_t}|^2$ be the probability that we obtain outcome $m=1$.
Assuming that $\theta$ is promised to be one of the two following cases:
\begin{outline}[enumerate]
\1 $\theta>\theta_i+\varphi_i$: The probability that we obtain outcome $m=1$ is
\begin{align}
\kappa
=\sum_{j}|\langle\psi_j|\psi_t\rangle|^2\frac{1}{2}\sum_{\pm}|\beta(\Theta_{j,\pm})|^2
\le
\sum_{j}|\langle\psi_j|\psi_t\rangle|^2\max_{j,\pm}|\beta(\Theta_{j,\pm})|^2=\max_{j,\pm}|\beta(\Theta_{j,\pm})|^2
\le
\delta_i.
\end{align}
\1 $\theta\le\theta_i-\varphi_i$: The probability that we obtain outcome $m=1$ is
\begin{align}
\kappa\nonumber
&=\sum_{j}|\langle\psi_j|\psi_t\rangle|^2\frac{1}{2}\sum_{\pm}|\alpha(\Theta_{j,\pm})|^2
=
\frac{p}{2}\sum_{\pm}(1-|\alpha(\pm\theta)|^2)+
\sum_{j>0}|\langle\psi_j|\psi_t\rangle|^2\frac{1}{2}\sum_{\pm}|\alpha(\Theta_{j,\pm})|^2\\
&\ge
p(1-\delta_i).
\end{align}
\end{outline}
Let us choose $\delta_i=p/4$. 
Hence, we can determine which of the two cases hold by deciding whether $\kappa\le\delta_i=p/4$ or $\kappa\ge p(1-\delta_i)\ge p(1-p/4)\ge3p/4$.
Now define the operator $O$ and state $\ket{\psi}$ to be
\begin{align}
O&:=(\ket{0}\bra{0}\otimes I)\textsc{Gpe}_{\varphi_i,\theta_i,\delta_i},\quad
\ket{\psi}=\ket{0}\ket{0}\ket{\psi_t},\quad
\Rightarrow\quad \kappa=\bra{\psi}O^\dagger O\ket{\psi}.
\end{align}
As we can block-encode $\Be[O/1]$ using a single controlled-$\textsc{Not}$ gate and $1$ query to $\textsc{Gpe}_{\varphi_i,\theta_i,\delta_i}$, we can solve this decision problem by applying either~\cref{lem:fast_estimation_single_term} or~\cref{lem:fast_estimation_single_term_no_prior} to find an estimate $\hat{\kappa}$ of $\kappa\le 3p/4$ to additive error $p/2$ and failure probability $q_i$ using $Q_{P,i}=\mathcal{O}\left(\frac{1}{\sqrt{p}}\log\frac{1}{q_i}\right)$ queries to $P$ and its inverse, and $Q_{H,i}=\mathcal{O}\left(\frac{1}{\sqrt{p}r^{i}}\log\frac{1}{p}\log\frac{1}{q_i}\right)$ queries to $\Be[H/{\lambda}]$.

Hence, we make the assignment
\begin{align}
\mathcal{I}_{i+1}=\begin{cases}
\mathcal{I}_{i,\downarrow},& \hat{\kappa}> p/4,\\
\mathcal{I}_{i,\uparrow}, &\hat{\kappa} \le p/4.
\end{cases}
\end{align}
The estimate of $\theta$ in $\mathcal{I}_{i}$ converts to an estimate of 
\begin{align}
\alpha&\in-\cos(\mathcal{I}_{i})=[-\cos(\mathcal{I}_{i,\text{l}}),-\cos(\mathcal{I}_{i,\text{r}})]
=
[-\cos(\theta_{i}-|\mathcal{I}_{i}|/2),-\cos(\theta_{i}+|\mathcal{I}_{i}|/2)],
\\\label{eq:fast_expectation_estimation_epsilon_1}
\epsilon_i&=|\cos(\mathcal{I}_{i})|=2\sin(\theta_i)\sin(|\mathcal{I}_i|/2),
\\
E+\lambda&\in\mathcal{H}_i=[\mathcal{H}_{i,\text{l}},\mathcal{H}_{i,\text{r}}],\quad \mathcal{H}_i:=\lambda(1-\cos\mathcal{I}_{i}),
\\
\epsilon_i'&:=|\mathcal{H}_i|=\lambda\epsilon_i,
\end{align}
where $\epsilon_i$ and $\epsilon_i'$ is the additive error to which $\alpha$ and $E$ is known respectively..

The remainder of the proof mirrors that of~\cref{lem:fast_estimation_single_term_no_prior} starting from~\cref{eq:fast_expectation_estimation_epsilon_1}.
We choose $q_i=\frac{6}{\pi^2}\frac{q/2}{(i_\text{max}-i+1)^2}$.
Then by a union bound, the failure probability after all $i_\text{max}$ steps is at most $q/2$ and we have identified with at least probability $1-q/2$ that $\theta\in\mathcal{I}_{i_\text{max}}$ with additive error at most $|\mathcal{I}_{i_\text{max}}|=\frac{\pi}{2}r^{i_\text{max}}$.
The total query complexity is then 
\begin{align}
Q_{H}&=\sum_{i=0}^{i_\text{max}-1}\mathcal{O}\left(\frac{1}{\sqrt{p}r^{i}}\log\frac{1}{p}\log\frac{1}{q_i}\right)
=\mathcal{O}\left(\left(\frac{2}{3}\right)^{i_\text{max}}\frac{1}{\sqrt{p}}\log\frac{1}{p}\log\frac{1}{q}\right),
\\
Q_P&=\sum_{i=0}^{i_\text{max}-1}\mathcal{O}\left(\frac{1}{\sqrt{p}}\log\frac{1}{q_i}\right)
=\mathcal{O}\left(\frac{i_\text{max}}{\sqrt{p}}\log\frac{i_\text{max}}{q}\right).
\end{align}
We now make two choices of $i_\text{max}$:
\begin{outline}[enumerate]
    \1 Let $i_\text{max}=i_{\text{max},1}+\Theta(1)$, where $i_{\text{max},1}=\frac{1}{2}\log_{1/r}\frac{\lambda}{\epsilon'}$. 
    Then we either determine that $E+\lambda\le\epsilon'$ and terminate the algorithm, or obtain an estimate $\hat{E}=\Theta(E)$ to constant multiplicative error. 
    \1 Let $i_\text{max}=i_{\text{max},2}+\Theta(1)$, where $i_{\text{max},2}=\log_{1/r}\frac{\sqrt{\hat{E}\lambda}}{\epsilon'}$. 
    Then we estimate $E$ to additive error at most $\epsilon'$.
\end{outline}
The overall query complexity for these two loops over different $i_{\text{max}}$ is
\begin{align}
Q_{H}&
=\mathcal{O}\left(\frac{\sqrt{\lambda}}{\sqrt{p}\epsilon'}(\sqrt{\epsilon'}+\sqrt{E+\lambda})\log\frac{1}{p}\log\frac{1}{q}\right),
\\
Q_P&
=\mathcal{O}\left(\frac{\log(\lambda/\epsilon^{\prime})}{\sqrt{p}}\log\frac{\log(\lambda/\epsilon')}{q}\right).
\end{align}
\end{proof}

\subsection{Optimality of spectral amplification}

Some results in \Cref{tab:results} are presented using $\Theta(.)$ notation and thus far we only commented on the upper bounds on the query complexity for quantum phase estimation and energy estimation. Indeed, these results can be shown to be optimal and tight lower bounds can be found. To this end, we use a similar problem than that used for the lower bounds in Ref.~\cite{zlokapa2024hamiltonian}, which essentially involves solving $K$ instances of unstructured search and then using a query lower bound for these. We note that similar lower bounds can be obtained using  related approaches based on polynomial approximations. 

\begin{definition}[$K$-partition ${{\rm PARITY} \circ {\rm OR}}$ decision problem]
    We denote the $K$-partition decision problem over a length-$KN$ bitstring by ${{\rm PARITY} \circ {\rm OR}}_{K,N}$, where $N=2^n$ and $K \ge 1, n \ge 1$ are integer. Each partition $k \in [K]:=\{1,\ldots,K \}$ consists of $N$ bits, where one bit is marked $f_k(x_k) = 1$ for some $x_k \in \{0, \ldots, N-1\}$, and all other bits $x \ne x_k$ are marked $f_k(x) = 0$. Access to the functions $f_1, \ldots , f_K : \{0, 1\}^N\mapsto \{0, 1\}$  is provided through a phase oracle $O_f$ that implements $\ket k \ket x \mapsto -1^{f_k(x)} \ket k \ket x$. Let $g:\{0,1\}^N \mapsto \{0,1\}$ be such that $g(x)$ is the OR of the first $N/2$ bits of $x$. The ${{\rm PARITY} \circ {\rm OR}}$ decision problem is to return YES if ${\rm PARITY}(g(x_1), \ldots, g(x_K)) = 1$ and NO otherwise, with probability at least 2/3.
\end{definition}
  
Reference~\cite{zlokapa2024hamiltonian} shows a lower bound $\Omega(K \sqrt N)$ on the number of uses of $O_f$. The result is intuitive from the $\Omega(\sqrt N)$ lower bound on Grover's search but the proof requires  some technical considerations that are addressed in that work. We now show how ${{\rm PARITY} \circ {\rm OR}}$ can be reduced to a phase estimation problem involving positive semidefinite Hamiltonians that can be used in conjunction with SA. Among other properties, these Hamiltonians can be used to solve $K$ instances of SEARCH. 

For a bitstring $x$, let 
\begin{align}
    H_x : = ( \one - \ketbra x - \ketbra s)^2 \; , \; \ket s:= \sqrt{\frac 2 N }\sum_{x'=0}^{N/2-1}\ket x' \;.
\end{align}
(The operator $\one$ is the identity of dimension $N$.)
Consider the two-dimensional subspace spanned by $\ket x$ and $\ket s$. In that subspace, the eigenvalues of $H_x$ are
\begin{align}
   \frac 2 N \ & {\rm if} \ \bra x s\rangle = \sqrt{2/N} \;, \\
    0 \ & {\rm if} \ \bra x s\rangle = 0 \;.
\end{align}
Furthermore, $\ket s$ is an eigenvector since the subspaces are doubly degenerate. We extend the definition to consider  $K$ subsystems in the following way:
\begin{align}
    H_{K,N}:=\sum_{k=1}^K \one^{\otimes k-1} \otimes H_{x_k} \otimes \one^{K-k} \;.
\end{align}
This Hamiltonian has the terms $H_{x_k}$ acting independently on the subsystems and $\ket s^{\otimes K}$ is an eigenvector of eigenvalue $\Delta \le 2K/N$. Since $\sqrt{H_{x_k}}=\one - \ketbra{x_k}-\ketbra s$, we can accordingly define the square root of $H_{K,N}$ for SA:
\begin{align}
    H^{\rm SA}_{K,N}:=\sum_{k=1}^K \ket k \otimes 
    \one^{\otimes k-1} \otimes (\one - \ketbra{x_k}-\ketbra s) \otimes \one^{K-k} \;,
\end{align}
and note that $H_{K,N} =   (H^{\rm SA}_{K,N})^\dagger   H^{\rm SA}_{K,N}$. Also, since $\|(\one - \ketbra{x_k}-\ketbra s)\|\le 1$, we obtain $\| H_{K,N}\|\le K$ and $\| H^{\rm SA}_{K,N}\|\le \sqrt K$.

\begin{theorem}
    Let $N\ge 1$, $K \ge 1$, $\lambda_K:=K$. Assume access to the block-encoding $\Be[H^{\rm SA}_{K,N}/\sqrt{\lambda_K}]$.  Let the initial low-energy eigenstate of $H_{K,N}$ be $\ket s^{\otimes K}$ that has eigenvalue at most $\Delta_{K,N}:=2K /N$. Then, performing quantum phase estimation on $H_{K,N}$ with precision $\epsilon_N= 1/ N$ requires $\Omega(\sqrt{\lambda_K \Delta_K}/\epsilon_N)$ calls to $\Be[H^{\rm SA}_{K,N}/\sqrt{\lambda_K}]$.
\end{theorem}

\begin{proof}
    The first step is to reduce $K$-partition ${\rm PARITY} \circ {\rm OR}$ to this quantum phase estimation problem. To this end we note that whenever the marked element is $x_k <N/2$ then this increments the eigenvalue of $H_{K,N}$ by $2/N$. Hence, by determining the eigenvalue of $H_{K,N}$ to within precision $1/N$ we can count the number of instances where $x_k <N/2$ and compute the parity of that string to solve $K$-partition ${\rm PARITY} \circ {\rm OR}$.

    The second step is to  note that $\Be[H^{\rm SA}_{K,N}/\sqrt{\lambda_K}]$ can be constructed with a constant number of calls to the oracle $O_f$. To this end, consider 
    \begin{align}
    \frac 1{\sqrt {K}} H^{\rm SA}_{K,N}:=  \frac 1{\sqrt {K}} \sum_{k=1}^K \ket k \otimes 
    \one^{\otimes k-1} \otimes (\one - \ketbra{x_k}-\ketbra s) \otimes \one^{K-k} \;.
    \end{align}
    A block-encoding can be prepared by first preparing the state
    \begin{align}
       {\rm PREPARE}\ket 0 \mapsto    \frac 1 {\sqrt K} \sum_{k=1}^K \ket k \;,
    \end{align}
    and then applying the SELECT unitary
    \begin{align}
        {\rm SELECT}:= \sum_{k=1}^K \ketbra k \otimes   \one^{\otimes k-1}  \otimes \Be[
          (\one - \ketbra{x_k}-\ketbra s) ] \otimes \one^{K-k}\;.
    \end{align}
    Since $\one - \ketbra{x_k}-\ketbra s$ is a linear combination of at most two unitaries, the block-encoding
    $\Be[(\one - \ketbra{x_k}-\ketbra s)]$ can be obtained
    with one evaluation of $f_k$:
    \begin{align}
         \one - \ketbra{x_k}-\ketbra s= \frac 1 2 \left(  (\one -2 \ketbra{x_k}) +   (\one -2 \ketbra{s}) \right) = \frac 1 2 \left(  \bra k O_f \ket k +   (\one -2 \ketbra{s}) \right) \;.
    \end{align}
    But since we have access to $O_f$ that computes $f_k$ on input $k$ (i.e., it is a conditional operation), we can construct the global ${\rm SELECT}$ with one use of $O_f$. The result is that $\Be[H^{\rm SA}_{K,N}/\sqrt{\lambda_K}] \equiv 
    \Be[{\rm SELECT}.{\rm PREPARE}\ket 0]$ only requires one use of $O_f$.
    
    The lower bound on $K$-partition ${\rm PARITY} \circ {\rm OR}$ now implies that at least $\Omega(K \sqrt N)$ uses of $\Be[H^{\rm SA}_{K,N}/\sqrt{\lambda_K}]$ are needed for quantum phase estimation in this example. This can be alternatively expressed as
    \begin{align}
        \Omega \left(\frac {\sqrt{\lambda_K \Delta_K}}{\epsilon_N} \right) = \Omega \left( \sqrt{K \times \frac{2K}N} \times N \right) = \Omega (K \sqrt N)
    \end{align}
    as desired.
\end{proof}

The same result automatically applies to expectation estimation, since quantum phase estimation is an example of the former. That is, the previous result is a case of estimating $\bra \phi H_{K,N} \ket \phi$ within precision $1/N$, where $\ket \phi \equiv \ket s^{\otimes K}$.

\subsection{Spectral amplification and linear combination of unitaries}

Thus far we provided the query complexities for various simulation tasks using spectral amplification. These results required access to $\Be[H_{\rm SA}/\sqrt \lambda]$, which can be constructed from the $\Be[A_j/ a_j]$ as explained in~\Cref{thm:sa}. In applications we will need to construct these block-encodings from some presentation of the Hamiltonian terms and here we discuss the LCU presentation within the context of SA. To this end, we will assume that the terms $h_j \succeq 0$ in Eq.~\eqref{eq:sa_form} can be expressed as $h_j = A^\dagger_j A^{}_j$, for some operators $A_j$ that are expressed as
\begin{align}
\label{eq:BjforLCUSA}
    A_j = \sum_{l=0}^{L-1} a_{jl} \sigma_{jl} \;,
\end{align}
where the coefficients $a_{jl} \ge 0$ without loss of generality and 
the $ \sigma_{jl}$'s are unitary, e.g., Pauli strings. We seek to construct
a block-encoding for $H_{\rm SA}$ from accessing these unitaries. We can obtain the following result.
\begin{lemma}
\label{lem:salcu}
    Let $H_{\rm SA}=\sum_{j=0}^{R-1} \ket j _{\rm a} \otimes A_j$ be as in Eq.~\eqref{eq:sga_def}, where $A_j$ is presented in Eq.~\eqref{eq:BjforLCUSA}.
    Then, we can implement a block-encoding $\Be[H_{\rm SA}/\sqrt{\lambda}]$, where
    \begin{equation}
        \lambda:= \sum_{j=0}^{R-1} (\|\vec a_j \|_1 )^2  \; , \; \|\vec a_j \|_1 := \sum_{l=0}^{L-1} a_{jl} \;.
    \end{equation}
    This requires $\cO(R \times L \times C)$
    gates, where $C$ is the gate cost of the $\sigma_{jl}$'s.
\end{lemma}
\begin{proof}
    The proof follows the steps of \Cref{thm:sa}. To implement SELECT, we need to replace $\Be[A_j/a_j]$ in that proof 
    by
    \begin{align}
        \Be \left [\frac{A_j} {\sum_{l=0}^{L-1}  a_{jl} } \right ]= \Be \left [\frac{\sum_{l=0}^{L-1} a_{jl}\sigma_{jl}} {\sum_{l=0}^{L-1}  a_{jl} } \right ] \;.
    \end{align}
    Note that we have effectively replaced $a_j$ in  \Cref{thm:sa} by $\sum_{l=0}^{L-1}  a_{jl} \equiv \|\vec a_j\|_1$ and $A_j$ by the LCU. In that proof, $\lambda=\sum_j a_j$ and this would give the normalization factor $\lambda =\sum_j (\|\vec a_j\|_1)^2$ in this case.

    The block-encoding can be constructed using standard techniques since it involves an LCU. We can define
    \begin{align}
        {\rm SELECT}_j : = \sum_{l=0}^{L-1} \ketbra l_{\rm c} \otimes \sigma_{jl}
    \end{align}
    and
    \begin{align}
        {\rm PREPARE}_j \ket 0_{\rm c}\mapsto \frac 1 {\sqrt{\sum_{l=0}^{L-1}  a_{jl} }} \sum_{l=0}^{L-1} \sqrt{a_{jl}} \ket l_{\rm c} \;,
    \end{align}
    so that
    \begin{align}
        \bra 0_{\rm c} ( {\rm PREPARE}_j)^\dagger \cdot {\rm SELECT} \cdot {\rm PREPARE}_j \ket 0_{\rm c} = \frac{\sum_{l=0}^{L-1} a_{jl}\sigma_{jl}} {\sum_{l=0}^{L-1}  a_{jl}} \;.
    \end{align}
    That is, $({\rm PREPARE}_j)^\dagger \cdot {\rm SELECT}\cdot {\rm PREPARE}_j$ gives the desired $\Be[{A_j}/ {\sum_{l=0}^{L-1}  a_{jl} } ]=\Be[{A_j}/ \|\vec a_{j}\|_1 ]$.
    The total gate cost is then dominated by the query cost of \Cref{thm:sa}, which is $\cO(R)$, times the gate cost of this block-encoding $\Be[{A_j}/ \|\vec a_{j}\|_1 ]$,
    which is $\cO(L \times C)$, where $C$ is the gate cost of the $\sigma_{jl}$.
\end{proof}

\section{SOS optimization and example SOSSA block-encoding}

In this section we provide necessary details for constructing the mathematical program that determines an SOS representation of a Hamiltonian given an operator basis for the SOS generators such that the lower bound energy is maximized. Recall, that maximizing the lower bound is an important algorithmic component in the SOSSA protocol that can lower the query complexity and end-to-end gate complexity if balanced with the cost of the block-encoding implementation. We also provide an example block-encoding for $H_{\sossa}$ expressed as a polynomial of Pauli operators.

\subsection{SOS optimization} \label{sec:sos}

Section~\ref{mainsec:sos} advocates for an SOS representation with increased complexity beyond termwise SA, which resulted in a loose lower bound that potentially negates an advantage through SA. In the following, we will restrict our exposition to Hamiltonians composed of Pauli strings but the construction can be applied more generally. We will also start with the formulation of the SOS optimization as a mathematical program and return later to the motivation of this form. Termwise SA uses a restricted algebra to form an SOS operator (in that case a projector) made from the linear combination of an identity operator and the Pauli operator of the Hamiltonian (Eq.~\eqref{eq:pauli_sos_projector}). The lower bound $-\beta$ on the ground state energy can be improved by providing more variational freedom in the SOS generators, the $B_{j}$ operators of Eq.~\eqref{eq:ham_sos_equality}. The variational protocol can be formulated as a semidefinite program (SDP). As an illustrative example, consider the set of operators constituting the span of one and two-qubit Pauli operators over $N$-qubits organized into a column vector
\begin{equation}
\vec{X}^{(2)} = \big( \one, X_1, \dots, Z_N, X_1 X_2, \dots, Z_{N-1} Z_N \big)^{T}.
\end{equation}
If our Hamiltonian is $2$-local, then this algebra is sufficient to represent the Hamiltonian. In the following, we consider $\vec{X}^{(k)}$ which involves operators involving degree-$k$ products of Pauli operators. A particular SOS generator is defined as $B_{j}\coloneqq \vec{b}_{j}\vec{X}^{(k)}$ for a row vector, $\vec b_j \in \mathbb{C}^{\binom{N}{k}}$, of complex coefficients to be optimized in order to satisfy Eq.~\eqref{eq:ham_sos_equality}. To show the optimization protocol is an SDP, we form the SOS representation of the Hamiltonian from a set of operators $\{B_{j}\}$:
\begin{align}\label{eq:sos_form_vec}
\sum_{j=0}^{R-1}B_{j}^{\dagger}B^{\!}_{j} &= \sum_{j=0}^{R-1}\left(\vec{b}_{j}\vec{X}^{(k)}\right)^{\dagger}\left(\vec{b}_{j}\vec{X}^{(k)}\right) \\
&=  \left(\vec{X}^{(k) \dag}\right)^{T} \Big(\sum_j \vec{b}_j^{\dag} \vec{b}_j \Big) \vec{X}^{(k)} \\
&= \left(\vec{X}^{(k) \dag}\right)^{T} G \vec{X}^{(k)} \;,
\end{align}
which can be used to  construct the program
\begin{align}\label{eq:app:sossdp}
&\min_{G}\;\beta \\
&\mathrm{s.t.}\;\; H + \beta \one = \left(\vec{X}^{(k) \dag}\right)^{T} G \vec{X}^{(k)} \nonumber \\
&G \succeq 0. \nonumber
\end{align}
The above SDP equality constraints are short-hand for relating the coefficients of the operators obtained by expanding the right-hand side of the equality to the left-hand side (the Hamiltonian and the shift). This computationally the equality constraint is represented by letting the Hamiltonian coefficients in matrix form, $\mathbf{H} \in \mathbb{C}^{L \times L}$ where $L = |\vec{X}^{(k)}|$ provide the constants resulting from the Hilbert-Schmidt inner products, $\langle \mathbf{A}|G\rangle$, between constraint matrix $\mathbf{A}$ and $G$ that encodes the coefficient relationship of the equality constraint. Similarly, the cost function corresponds to minimizing the coefficient associated with the identity operators, which in this case corresponds to the diagonal elements of the Gram matrix $G$. The SDP of Eq.~\eqref{eq:app:sossdp} can be solved in polynomial time with respect to the linear dimension of $G$. Once the SDP is solved an SOS representation can be recovered by expressing $G$ in its eigenbasis $G = \sum_j \mu_j \vec{a}_j \vec{a}_j^\dag$, where $\{\vec{a}_j\}$ which yields a definition for each $B_{j}$
\begin{equation} \label{eq:gram}
\vec{X}^{(k) \dag} G \vec{X}^{(k)} = \sum_{j=0}^{R-1} (\sqrt{\mu_j} \vec{a}_j^\dag \vec{X}^{(k)})^\dag (\sqrt{\mu_j} \vec{a}_j^\dag \vec{X}^{(k)}) = \sum_{j=0}^{R-1}B_{j}^{\dagger}B^{\!}_{j}.
\end{equation}
The number of $B_{j}$'s in the generating polynomials is equal to the rank $R$ of the Gram matrix $G$.

The intuition behind providing more variational freedom through the form of the SOS generator is that for any $-\beta$ that is less than the ground state energy $E$ of some Hamiltonian $H$ we can define their sum $\tilde{H} = H + \beta \one$ to be a positive semidefinite operator whose square root $\tilde{H} = \sqrt{\tilde{H}}\sqrt{\tilde{H}}$ which is potentially a different, and potentially more complicated, many-body body operator. Thus if the equality of Eq.~\eqref{eq:ham_sos_equality} is satisfied then the SOS Hamiltonian serves as a certificate that the energy is lower bounded by $-\beta$.  The certificate perspective is the primary method for constructing constraints on pseudomoment matrices in `outer' approximation methods that seek to minimize representations of marginals and approximately constrain them through the of sum-of-square construction~\cite{PhysRevA.75.032102, wittek2015algorithm, erdahl1978representability, nakata2001variational}.  In fact, it can be shown that the SOS Hamiltonians can be determined from the primal problem directly~\cite{low2025fast}.

\subsection{SOSSA block-encoding for Pauli operators} \label{sec:sossa}

Utilizing the solution of the SDP to construct the SOS generators  $B_{j}$, we can now replace these in the spectral amplified Hamiltonian; that is, we can replace $A_j \rightarrow B_j$ in Eqs.~\eqref{eq:sga_def},~\eqref{eq:samainproperty}.
Expanding Eq.~\eqref{eq:gram}, each SOS generator is a linear combination of Pauli operators
\begin{equation}
\label{eq:bjpaulis}
B_j = \sum_{l=0}^{L-1} b_{jl} \sigma_l \;,
\end{equation}
where $\{\sigma_l\}$ is the monomial basis and $L$ is the number of monomials in the SOS optimization. To block encode $H_{\sossa}$ of Eq.~\eqref{maineq:sossa_sqrt}, which is essentially $H_{\rm SA}$ in Eq.~\eqref{eq:sga_def} after replacing $A_j \rightarrow B_j$, gives us the following block-encoding normalization factor.
\begin{definition}[SOS $\lambda$] \label{def:sos_lambda}
    Define $\lambda_{\sos}$ to be
    \begin{equation}
        \lambda_{\sos} := \sum_{j=0}^{R-1} \|\vec{b}_j\|_1^2 \;, \; \|\vec b_j\|_1=\sum_{l=0}^{L-1} b_{jl} \;.
    \end{equation}
\end{definition}

This is the normalization factor in~\Cref{lem:salcu} for this choice of operators $B_j$, which readily implies the following.
\begin{corollary}[Compilation of SOSSA with Pauli strings] \label{thm:sossa}
    Let $H_{\sossa}=\sum_{j=0}^{R-1}\ket  j \otimes B_j$, where the $B_j$'s are in Eq.~\eqref{eq:bjpaulis}. 
    Then, we can construct the block-encoding  $\Be[H_{\sossa} / \sqrt{\lambda_{\sos}}]$
    with  $\mathcal{O}(R \times L)$ gates.
\end{corollary}
Since the Pauli strings are of constant weight, $C$ is a constant in ~\Cref{lem:salcu}.

\vspace{0.1cm}

In \Cref{app:sec:SYK} we rely on upper and lower bounds on $\lambda_{\sos}$. To derive these, consider a general Hamiltonian $H \in \mathbb C^{M \times M}$ such that after the shift it produces $H+\beta \one =\sum_{j=0}^{R-1} B^\dagger_j B^{\!}_j$ as in Eq.~\eqref{eq:ham_sos_equality}.
Pauli strings are orthogonal and their trace is zero, implying
\begin{equation}
\sum_{j=0}^{R-1} \|\vec{b}_j\|_2^2 = \frac 1 M {\Tr}{H} + \beta \;.
\end{equation}
This gives us bounds of $\lambda_{\sos}$ via Cauchy-Schwarz:
\begin{equation}
\frac 1 M{\Tr}{H} + \beta \ \leq \ \lambda_{\sos} \ \leq \ L   \left(\frac 1 M {\Tr}{H} + \beta \right).
\end{equation}

\section{Application to SYK model}\label{app:sec:SYK}

In this section we describe the application of the SOSSA framework to the SYK model. We will see a factor of $\sqrt{N}$ speedup in both query and gate complexities when compared to the standard LCU compilation.  This is accomplished by showing that $\Delta_{\sos}$ scales linearly in system size with high probability and $\lambda_{\sos}$ scales quadratically. 

Let $\{\gamma_1,\dots,\gamma_N\}$ be Majorana operators satisfying anticommutation relations:
\begin{equation}
    \gamma_a \gamma_b + \gamma_b \gamma_a = 2 \delta_{ab} \one \;.
\end{equation}
The SYK model is described by a fermionic Hamiltonian containing all degree-4 Majorana terms whose coefficients are random Gaussians where
\begin{align}
H_{\syk} &= \frac{1}{\sqrt{N \choose 4}} \sum_{a,b,c,d} g_{abcd} \gamma_a \gamma_b \gamma_c \gamma_d \;, \\
g_{abcd} &\sim \mathcal{N}(0,1) \ \text{i.i.d.} \;.
\end{align}

To implement SOSSA, we will use the degree-2 Majorana SOS, which we describe in \Cref{sec:sos_Majorana_2}. We will then compile the resulting SOSSA Hamiltonian using the double factorization technique, which we recap in \Cref{sec:double_factorization}. In \Cref{sec:syk}, we use random matrix theory to analyze the energy gap $\Delta_{\sos}$ of the degree-2 Majorana SOS on the SYK model, and we also analyze the normalization factor $\lambda_{\sos}$ resulting from the double factorization block encoding. Putting these together demonstrate our asymptotic improvements.

\subsection{Degree-2 Majorana SOS} \label{sec:sos_Majorana_2}

For fermionic systems, a natural SOS ansatz is the degree-2 Majorana SOS, which we will use for the SYK model. In this section we describe how to implement degree-2 Majorana SOS for a general fermionic Hamiltonians of the form:
\begin{equation}
H = i \sum_{a,b=1}^N K_{ab} \gamma_a \gamma_b - \sum_{a,b,c,d=1}^N J_{abcd} \gamma_a \gamma_b \gamma_c \gamma_d
\end{equation}
Our $\sos$ is defined by the basis $\{\one, \gamma_a, \gamma_a \gamma_b\}$. Using \Cref{sec:sos}, we can write the SOS relaxation as ($\beta \in \mathbb R$)
\begin{align}
\sos(H) \quad = \qquad & \min{\beta}  \\
\text{s.t.} \quad & H + \beta \one = \vec{X}^\dag G \vec{X} \nonumber\\
& G \succeq 0 \nonumber
\end{align}
where
\begin{equation}
\vec{X} = \big(\one, \gamma_1, \dots, \gamma_N, i \gamma_1 \gamma_2, \dots, i \gamma_{N-1} \gamma_N\big)
\end{equation}
and $G$ is the Gram matrix with dimension $1 + N + {N \choose 2}$. Comparing the coefficients of $\one$ in the polynomial constraint equation reveals that $\beta = \Tr{G}$, since all other operators in the generating set are traceless. Comparing coefficients of the other monomials gives us the linear constraints on the SDP. Thus we can explicitly write the SDP
\begin{align}
\min{\Tr{G}} \\
\text{s.t.} \quad & 0 = G(\one, \gamma_a) + G(\gamma_a, \one) + i \sum_c G(i \gamma_a \gamma_c, \gamma_c) - i \sum_c G(\gamma_c, i \gamma_a \gamma_c) & \forall a \nonumber\\
& K_{ab} = G(\one, i \gamma_a \gamma_b) + G(i \gamma_a \gamma_b, \one) - i G(\gamma_a, \gamma_b) + i G(\gamma_b, \gamma_a) \nonumber\\
&\qquad - i \sum_c G(i \gamma_a \gamma_c, i \gamma_b \gamma_c) + i \sum_c G(i \gamma_b \gamma_c, i \gamma_a \gamma_c) & \forall a,b \nonumber\\
& 0 = G(\gamma_a, i \gamma_b \gamma_c) + G(\gamma_b, i \gamma_c \gamma_a) + G(\gamma_c, i \gamma_a \gamma_b) \nonumber\\
&\qquad + G(i \gamma_b \gamma_c, \gamma_a) + G(i \gamma_c \gamma_a, \gamma_b) + G(i \gamma_a \gamma_b, \gamma_c) & \forall a,b,c \nonumber\\
& J_{abcd} = G(i \gamma_a \gamma_b, i \gamma_c \gamma_d) + G(i \gamma_c \gamma_d, i \gamma_a \gamma_b) \nonumber\\
&\qquad - G(i \gamma_a \gamma_c, i \gamma_b \gamma_d) - G(i \gamma_b \gamma_d, i \gamma_a \gamma_c) \nonumber\\
&\qquad + G(i \gamma_a \gamma_d, i \gamma_b \gamma_c) + G(i \gamma_b \gamma_c, i \gamma_a \gamma_d) & \forall a,b,c,d \nonumber\\
& G \succeq 0 \nonumber
\end{align}

\subsection{Double factorization} \label{sec:double_factorization}

The degree-2 Majorana SOS leads to a representation of the Hamiltonian that takes the form
\begin{align}
&H + \beta\one = \sum_j B_j^\dag B_j^{\!} \;,\\
&B_j^{\!} = e_j \one + \sum_a f_{j,a} \gamma_a + \sum_{ab} g_{j,ab} \gamma_a \gamma_b
\end{align}
A similar factorization occurs in the 2-body term of quantum chemistry Hamiltonians. Using the direct block encoding strategy, \Cref{def:sos_lambda} gives for this SOS representation a SOS $\lambda$ of
\begin{equation} \label{eq:SF_lambda_sos}
\lambda_{\sos} = \sum_{j=0}^{R-1} \Big(|e_j| + \|f_j\|_1 + \sum_{ab} |g_{j,ab}|\Big)^2
\end{equation}
where $R$ is the rank of the Gram matrix $G$.

We can achieve a much better $\lambda_{\sos}$ using the concept of double factorization from quantum chemistry \cite{von2021quantum}, which we now describe. See \cite{lee2021even} for a comprehensive discussion of quantum algorithms techniques for quantum chemistry. Given a quadratic polynomial in the Majorana operators $\sum_{ab} g_{ab} \gamma_a \gamma_b$, let's first decompose $g_{ab}$ into it's real and imaginary parts $g_{ab} = g^{(R)}_{ab} + i g^{(I)}_{ab}$. Without loss of generality, $g^{(R)}_{ab}$ and $g^{(I)}_{ab}$ are real antisymmetric, since $\gamma_a$ and $\gamma_b$ always anticommute. There are Gaussian unitaries $U^{(R)}, U^{(I)}$ which rotate the quadratic polynomials into the block-diagonal forms
\begin{align}
    \Big(\sum_{ab} g^{(R)}_{ab} \gamma_a \gamma_b\Big) &= {U^{(R)}}^\dag \sum_a \tilde{g}^{(R)}_a \gamma_{2a-1} \gamma_{2a} U^{(R)} \\
    \Big(\sum_{ab} g^{(I)}_{ab} \gamma_a \gamma_b\Big) &= {U^{(I)}}^\dag \sum_a \tilde{g}^{(I)}_a \gamma_{2a-1} \gamma_{2a} U^{(I)}
\end{align}
This is because a Gaussian unitary $U$ acts on the $N$-vector of Majorana operators $(\gamma_1,\dots,\gamma_N)$ via a real orthogonal matrix $O$, and we can always design $U$ so that $O$ block diagonalizes a given antisymmetric matrix $g_{ab}$:
\begin{equation}
\begin{pmatrix}
g_{11} & g_{12} & \dots & g_{1,N} \\
g_{21} & g_{22} & & \vdots \\
\vdots & & \ddots & \\
g_{N,1} & \dots & & g_{N,N}
\end{pmatrix}
=
O^T \begin{pmatrix}
0 & g_1 & & & \\
-g_1 & 0 & & & \\
 & & \ddots & & \\
& & & 0 & g_N \\
& & & -g_N & 0 \\
\end{pmatrix} O
\end{equation}

Let's find such Gaussian unitaries $U^{(R)}_j,U^{(I)}_j$ for each $g_{j,ab}$. This let's us write $B_j$ as
\begin{equation}
B_j = e_j \one + \sum_a f_{j,a} \gamma_a + {U^{(R)}}^\dag \sum_a \tilde{g}^{(R)}_{(j,a)} \gamma_{2a-1} \gamma_{2a} U^{(R)} + {U^{(I)}}^\dag \sum_a i \tilde{g}^{(I)}_{(j,a)} \gamma_{2a-1} \gamma_{2a} U^{(I)}
\end{equation}
We can efficiently implement the Gaussian unitaries $U_j$ on the quantum computer. Thus we can follow \Cref{def:sos_lambda} to derive $\lambda_{\sos}$ for the double factorized SOS representation, and the proof of \Cref{thm:sossa} will go through. We get
\begin{equation} \label{eq:DF_lambda_sos}
\lambda_{\sos} = \sum_{j=0}^{R-1} \big(|e_j| + \|f_j\|_1 + \|\tilde{g}^{(R)}_j\|_1 + \|\tilde{g}^{(I)}_j\|_1\big)^2
\end{equation}
This could be significantly cheaper than directly block encoding $B_j$, which would have $\lambda_{\sos}$ as give in \Cref{eq:SF_lambda_sos}.

\subsection{Application to SYK model} \label{sec:syk}

We can now bring all of these ingredients together to analyze the performance of degree-2 Majorana SOSSA with double factorization on the SYK model, when compared to the standard LCU approach.

The asymptotic gate complexities to construct the necessary block-encodings for LCU and SOSSA are both $\cO(N^4)$. First, notice that the number of terms in $H_{\syk}$ is $\sim N^4$, with each coefficient independently random, and each of the terms are independent. This creates an $\cO(N^4)$ gate cost for block encoding $H_{\syk}$ using known methods based on \textsc{SELECT} and \textsc{PREPARE} unitaries which is expected to be optimal. 
Double factorization also allows for a $\mathcal{O}(N^4)$ gate cost: the cost of the $B_j$ are dominated by the costs of the Gaussian unitaries, which are $\mathcal{O}(N^2)$, and the number of $B_j$ is at most $\mathcal{O}(N^2)$. The asymptotic equivalence of the block encoding gate costs means that we can focus on comparing the query complexities $\lambda_{\lcu}$ and $\sqrt{\Delta_{\sos} \lambda_{\sos}}$.

For the SYK model,
\begin{align}
    \lambda_{\lcu} = \frac 1 {\sqrt{\binom{N}{4}}}\sum_{a,b,c,d} |g_{abcd}|
\end{align}
which scales like $\lambda_{\lcu} \sim N^2$. We will see that $\sqrt{\Delta_{\sos} \lambda_{\sos}} \sim N^{\frac{3}{2}}$, giving a factor of $\sqrt{N}$ improvement in the query complexity.

We begin by analyzing the scaling of $\Delta_{\sos}$. This will depend on the quality of the SOS lower bound $-\beta$ arising from the degree-2 Majorana SOS described in \Cref{sec:sos_Majorana_2}. The SYK Hamiltonian is traceless, $\Tr(H_{\syk}) = 0$, implying that the ground state energy is negative, and it is known that with high probability the minimum and maximum eigenvalues 
are proportional to $\sqrt N$~\cite{hastings2022optimizing}.
The following lemma shows that applying degree-2 Majorana SOS to the SYK models gives a lower bound $-\beta$, where  $\beta= \mathcal{O}(N)$, giving an energy gap of $\Delta_{\sos} = \mathcal{O}(N)$.
\begin{lemma} \label{lem:syk_sos} 
Degree-2 Majorana SOS achieves a lower bound of $-\beta$ on the ground energy  of $H_{\syk}$, for $\beta= \mathcal{O}(N)$,  with high probability.
\end{lemma}
\begin{proof}
Our proof strategy is to use the dual problem to the SOS optimization. The SOS optimization obeys Slater's condition and thus exhibits strong duality -- see Section 5.2.3 of \cite{boyd2004convex}. Thus, in order to show the lower bound where $\beta = \mathcal{O}(N)$, it is sufficient to show that the lowest achievable energy in the dual problem is at least $-cN$ for some constant $c>0$.  That is, no solution to the dual problem can take a value lower than $-cN$.

Appendix B of \cite{pironio2010convergent} allows us to express the dual optimization problem to the $\sos$. Define a degree-2 pseudoexpectation to be a matrix $\tilde{\rho}$ whose rows and columns are indexed by quadratic Majorana operators, and where we impose that $\tilde{\rho}$ is positive semi-definite and obeys algebraic constraints
\begin{align}
&\tilde{\rho}(i\gamma_a\gamma_b,i\gamma_c\gamma_d) = -\tilde{\rho}(i\gamma_a\gamma_c,i\gamma_b\gamma_d) = \tilde{\rho}(i\gamma_a\gamma_d,i\gamma_b\gamma_c) \\
= &\tilde{\rho}(i\gamma_b\gamma_c,i\gamma_a\gamma_d) = -\tilde{\rho}(i\gamma_b\gamma_d,i\gamma_a\gamma_c) = \tilde{\rho}(i\gamma_c\gamma_d,i\gamma_a\gamma_b).
\end{align}
Let $J$ be the matrix with rows/columns indexed by subsets of $\{1,\dots,N\}$ of size 2, with entries
\begin{equation}
J(ab,cd) = \frac{1}{\sqrt{{N \choose 4}}} g_{abcd}
\end{equation}
where $g_{abcd}$ are the SYK coefficients. The dual problem is to minimize $\Tr(J \tilde{\rho})$ over degree-2 pseudoexpectations $\tilde{\rho}$.

To complete the proof, we show that any degree-2 pseudoexpectation $\tilde{\rho}$ obeys $\Tr(J \tilde{\rho}) \geq -cN$ with high probability. We have
\begin{equation}
\Tr(J \tilde{\rho}) \geq -\|J\|_{} \cdot \|\tilde{\rho}\|_1 = -\|J\|_{} \cdot \Tr(\tilde{\rho}) = -{N \choose 2} \|J\|_{} \;,
\end{equation}
where $\|.\|$ is the operator norm.
The first inequality was an application of matrix Holder's inequality \cite{baumgartner2011inequality}. Random matrix theory tells us that $\|J\|_{} = \mathcal{O}(N^{-1})$ with high probability. For example, see Theorem 4.4.5 in Ref.~\cite{vershynin2018high}. Thus $\Tr(J \tilde{\rho}) \geq -c N$ for any degree-2 pseudoexpectation $\tilde{\rho}$ with high probability over the SYK disorder.
\end{proof}

Next we analyze $\lambda_{\sos}$, which will involve the double factorization block encoding technique described in \Cref{sec:double_factorization}. The degree-2 Majorana SOS with double factorization writes the SYK Hamiltonian as
\begin{align}
&H_{\syk} + \beta \one = \sum_j B_j^\dag B_j \\
&B_j = e_j \one + \sum_a f_{j,a} \gamma_a + {U^{(R)}}^\dag \sum_a \tilde{g}^{(R)}_{(j,a)} \gamma_{2a-1} \gamma_{2a} U^{(R)} + {U^{(I)}}^\dag \sum_a i \tilde{g}^{(I)}_{(j,a)} \gamma_{2a-1} \gamma_{2a} U^{(I)}
\end{align}
where the $U_j$ are Gaussian unitaries. From \Cref{eq:DF_lambda_sos}, we can calculate
\begin{align}
\lambda_{\sos} &= \sum_j \big(|e_j| + \|f_j\|_1 + \|\tilde{g}^{(R)}_j\|_1 + \|\tilde{g}^{(I)}_j\|_1\big)^2 \\
&\leq 4N \Big(\sum_j |e_j|^2 + \sum_j \|\tilde{f}_j\|_2^2 + \sum_j \|\tilde{g}^{(R)}_j\|_2^2 + \sum_j \|\tilde{g}^{(I)}_j\|_2^2\Big) & \text{(Cauchy-Schwarz)}\\
&\leq 4N \cdot \overline{\Tr}\Big(\sum_j B_j^\dag B_j\Big) \\
&= 4N \cdot \overline{\Tr}\big(H_{\syk} + \beta \one\big) \\
&= 4N \beta \\
&= \mathcal{O}(N^2) & \text{(\Cref{lem:syk_sos})}
\end{align}
where $\overline{\Tr}$ denotes the normalized trace so that $\overline{\Tr}(\one) = 1$. Putting this together with the above analysis showing $\Delta_{\sos} = \mathcal{O}(N)$ gives $\sqrt{\Delta_{\sos} \lambda_{\sos}} = \mathcal{O}(N^{\frac{3}{2}})$.

\section{Bonus: Low-depth amplified amplitude estimation}\label{sec:low_depth_fast_expectation}
In this section, we use spectral amplification the reduce the sample complexity of estimating to standard deviation $\sigma$, the expectations of observables $E=\bra{\psi}H\ket{\psi}$ represented by a sum-of-squares $H=\sum_{j=0}^{R-1}A^\dagger_j A_j$~\cref{eq:samainproperty}, using circuits of very low-depth, and given copies of $\ket{\psi}$.
The expectation estimation results of in~\cref{sec:sa_energy_estimation} perform phase estimation using circuits with query and depth complexity $\mathcal{O}(\sqrt{\Delta\lambda}/\epsilon)$. 
Here, we instead perform Hadamard tests on the block-encoding $\textsc{Be}\left[\frac{A^\dagger_j A_j}{\frac{1}{2}|a_j|^2}-\mathbbm{1}\right]$, which can be formed by~\cref{lem:sos_block_encoding} using only one query to $\textsc{Be}[A_j/a_j]$ and its inverse.
This leads worse scaling with accuracy, but the low-depth circuits make this application more relevant to near-term applications.
By estimating terms with smaller normalization $a_j$ to lower precision and thus fewer samples, previous approaches would have a sample complexity $\cO\left(\left(\sum_{j=0}^{R-1}{|a_j|^2}/\sigma\right)^2\right)$~\cite{rubin2018marginals}.
However, using prior knowledge of upper bounds on the expectation of each term, say $\Delta_{j}\ge\bra{\psi}A^\dagger_jA_j\ket{\psi}$, SA allows us to achieve a sample complexity of  $\cO\left(\left(\sum_{j=0}^{R-1}{\sqrt{\Delta_j}|a_j|}/\sigma\right)^2\right)$ as follows from the following theorem.
\begin{theorem}[Expectation estimation by sums of squares at the standard quantum limit]\label{thm:low_depth_fast_expectation}
Let $H$ be any operator with a SOS decomposition $H=\sum_{j=0}^{R-1}A_j^\dagger A_j$.
Let $\Be[A_j/a_j]$ be a block-encoding of $A_j$.
Let each $\Delta_j\ge\bra{\psi}A^\dagger_j A_j|\psi\rangle=\langle H_j\rangle$, be any upper bound on the respective expectations.
Then the expectation $\bra{\psi}H|\psi\rangle$ can be estimated with variance at most $\sigma^2$ and no bias by repeating $N_j$ times the Hadamard test on $\Be[2A^\dagger_j A_ja^{-2}-\one]$ for each $j$ for total number of times
\begin{align}
\mathrm{Cost}=\sum_{j=0}^{R-1} N_j=2\left(\frac{\sum_{j=0}^{R-1}\sqrt{a_j^2\Delta_j}}{\sigma}\right)^2.
\end{align}
\end{theorem}
\begin{proof}
Let $\Lambda_j:=\frac{1}{2}a_j^2$, $\phi_j:=\langle H_j\rangle/\Lambda_j$, and $p_j:=1-(1-\phi_j)^2$.
Using one query to $\Be[A_j/a_j]$ and $\Be[A_j/a_j]^\dagger$, we can construct the controlled block-encoding
\begin{align}
\textsc{CBe}_j:=\ket{0}\bra{0}\otimes I+\ket{1}\bra{1}\otimes \Be_j,
\end{align}
where $\textsc{Be}_j=\Be\left[\frac{A_j^\dagger A_j}{\Lambda_j}-\one\right]$ is self-inverse.

The Hadamard test applies a controlled unitary $U$ on the $\ket{+}$ and $\ket{\psi}$ states to produce
\begin{align}
( H\otimes \one)U\ket{+}\ket{\psi}
=\frac{1}{2}\left(\ket{0}(I+U)\ket{\psi}+\ket{1}(I-U)\ket{\psi}\right).
\end{align}
The probability of measuring $\ket{0}$ is $p=\frac{1}{4}\bra{\psi}(I+U)^\dagger(I+U)\ket{\psi}=\frac{1}{2}(1+\bra{\psi}\mathrm{Re}[U]\ket{\psi})$.
Now choose $U$ to be $\Be\left[\frac{A_j^\dagger A_j}{\Lambda_j}-\one\right]$ and input state to be $\ket{+}\ket{0}\ket{\psi}$. 
Then we obtain $\ket{0}$ with probability
\begin{align}\nonumber
p_{j}&=\frac{1}{2}\left(1+\bra{0}\bra{\psi}\mathrm{Re}\left[\Be\left[\frac{A_j^\dagger A_j}{\Lambda_j}-\one\right]\right]\ket{0}\ket{\psi}\right)\\\nonumber
&=\frac{1}{2}\left(1+\bra{0}\bra{\psi}\left[\Be\left[\frac{A_j^\dagger A_j}{\Lambda_j}-\one\right]\right]\ket{0}\ket{\psi}\right)\\\nonumber
&=\frac{1}{2}\left(1+\bra{0}\bra{\psi}\left[\frac{A_j^\dagger A_j}{\Lambda_j}-\one\right]\ket{0}\ket{\psi}\right)\\\nonumber
&=\frac{1}{2}\bra{\psi}\left[\frac{A_j^\dagger A_j}{\Lambda_j}\right]\ket{\psi}
\\
&=\frac{1}{2}\phi_j
\end{align}
Let $X_{j,x}\sim \mathcal{B}(p_j)$ be identical independently distributed random variables following the Bernouli distribution.
If we repeat $N_j$ times, let us estimate $\phi_j$ using the estimator $\hat{\phi}_j=\frac{2}{N_j}\sum_{x\in[N_j]}X_{j,x}$.
This estimator has expectation $\mathbb{E}[\hat{\phi}_j]=\phi_j$ and variance $\mathbb{E}[(\hat{\phi}_j-\mathbb{E}[\hat\phi_j])^2]=\frac{4}{N_j} p_j(1-p_j)=\frac{2}{N_j} \phi_j(1-\frac{1}{2}\phi_j)\le\frac{2}{N_j} \phi_j.$.

Let us estimate $\bra{\psi}H|\psi\rangle$ by the sum $\langle\hat{H}\rangle=\sum_{j=0}^{R-1}\Lambda_j\hat{\phi}_j$.
Then this estimator has expectation and variance
\begin{align}
\langle\hat{H}\rangle
&=\sum_{j=0}^{R-1}\Lambda_j{\phi}_j
=\bra{\psi}H|\psi\rangle,\\
\langle(\hat{H}-\langle H\rangle)^2\rangle
&=
\sum_{j=0}^{R-1} \Lambda^2\mathrm{Var}[\hat{\phi}_j]
\le\sum_{j=0}^{R-1}\frac{4\Lambda_j^2 \phi_j}{N_j}
=\sum_{j=0}^{R-1}\frac{4\Lambda_j \langle H\rangle}{N_j}
\end{align}
Let us now choose $N_j$ so that the variance is bounded by $\sigma^2$.
We also want to choose $N_j$ in a way that minimizes the overall number of measurements $\mathrm{Cost}=\sum_{j=0}^{R-1}N_j$.
The optimal $N_j$ is obtained using a Lagrange multiplier $\lambda$ with the constraint $\sum_{j=0}^{R-1}\frac{4\Lambda_j\langle H_j\rangle}{N_j}=\sigma^2$.
However, we do not know {\textit a priori} the value of $\langle H_j\rangle$.
Hence, we use the upper bound $\langle H_j\rangle\le \Delta_j$ to determine $N_j$ and extremize $F=\sum_{j=0}^{R-1}N_j+\lambda(\sum_{j=0}^{R-1}\frac{4\Lambda_j \Delta_j}{N_j}-\sigma^2)$.
At the stationary point,
\begin{align}
\frac{\partial F}{\partial N_j}&=1-\lambda\frac{4\Lambda_j \Delta_j}{N_j^2}=0&\Rightarrow\quad N_j&=\sqrt{\lambda}\sqrt{{4\Lambda_j \Delta_j}},\\
\sigma^2&=\frac{1}{\sqrt{\lambda}}\sum_{j=0}^{R-1}\sqrt{4\Lambda_j \Delta_j}
&\Rightarrow\sqrt{\lambda}&=\frac{1}{\sigma^2}\sum_{j=0}^{R-1}\sqrt{4\Lambda_j \Delta_j},
\\
\mathrm{Cost}&=\frac{1}{\sigma^2}\left(\sum_{j=0}^{R-1}\sqrt{4\Lambda_j \Delta_j}\right)^2
=\lambda\sigma^2.
\end{align}
Hence, the optimal
\begin{align}
N_j=\frac{\sqrt{\mathrm{Cost}}}{\sigma}\sqrt{4\Lambda_j \Delta_j}=\frac{4}{\sigma^2}\sqrt{\Lambda_j \Delta_j}\sum_{j=0}^{R-1}\sqrt{\Lambda_j \Delta_j}.
\end{align}
\end{proof}

\end{document}